\documentclass[envcountsame]{llncs}

\usepackage[utf8]{inputenc}
\usepackage[T1]{fontenc}

\usepackage{xifthen}

\usepackage{amssymb}
\usepackage{amsmath}

\usepackage{xspace}

\usepackage{tikz}
\usetikzlibrary{matrix,arrows,positioning,trees,calc,shapes.geometric}
\usetikzlibrary{fit}

\usepackage{hyperref}

\hypersetup{
colorlinks=true,
pdfauthor={Alexander Heussner, Alexander Kartzow},
pdftitle={Reachability in Higher-Order-Counters},
pdfkeywords={}
}

\usepackage[final]{listings}
\definecolor{darkviolet}{rgb}{0.5,0,0.4}
\definecolor{darkgreen}{rgb}{0,0.4,0.2}
\definecolor{darkblue}{rgb}{0.1,0.1,0.9}
\definecolor{darkgrey}{rgb}{0.5,0.5,0.5}
\definecolor{lightblue}{rgb}{0.4,0.4,1}
\lstdefinestyle{eclipsish}{
    basicstyle=\scriptsize\ttfamily,
    emphstyle=\color{red}\bfseries,
    keywordstyle=\color{darkgreen}\bfseries,
    keywordstyle=[2]\color{darkviolet}\bfseries,
    commentstyle=\color{darkgrey},
    stringstyle=\color{darkblue},
    numberstyle=\color{darkgrey}\ttfamily\tiny,
    emphstyle=\color{red},
    morecomment=[s][\color{lightblue}]{/**}{*/},
  showstringspaces=Sfalse,
  numbers=left,
  numbersep=5pt,
  xleftmargin=2.5ex,
  xrightmargin=2.4ex,
  breakindent=3ex,
  breakautoindent,
  numberblanklines=false,
  escapeinside={(*@}{@*)},
  mathescape=true,
}

\lstdefinelanguage{pseudo}[]{C}
  {morekeywords=[2]{write},%
   alsoletter={^},%
   morekeywords= {def,foreach, elseif}
  }%
\newcommand{\PTIME}{\ensuremath{\mathbf{P}}\xspace}
\newcommand{\PSPACE}{\ensuremath{\mathbf{PSPACE}}\xspace}

\newcommand{\NSPACE}{\mathbf{NSPACE}}
\newcommand{\DSPACE}[0]{\mathbf{DSPACE}}
\newcommand{\DTIME}[0]{\mathbf{DTIME}}
\newcommand{\SPACE}[1]{\ensuremath{\mathbf{SPACE}({#1})}\xspace}

\newcommand{\wrapper}[2][]{%
\ifthenelse{\isempty{#1}}{%
\ensuremath{\mathsf{#2}}\xspace}
{\ensuremath{\mathsf{#1\text{-}#2}}\xspace}}

\newcommand{\HOPS}[1][]{\wrapper[#1]{HOPA}}
\newcommand{\HOCS}[1][]{\wrapper[#1]{HOCA}}
\newcommand{\HOCSO}[1][]{\wrapper[#1]{HOCA^{-}}}
\newcommand{\HOCSM}[1][]{\wrapper[#1]{HOCA^{+}}}

\newcommand{\Storage}{\ensuremath{\mathcal{S}}\xspace} 
\newcommand{\SConfig}{\ensuremath{X}\xspace}           
\newcommand{\Sconfig}{\ensuremath{x}\xspace}           
\newcommand{\SConfigInit}{\ensuremath{x_0}\xspace}     

\newcommand{\CStorage}{\ensuremath{\mathcal{Z}}\xspace}      
\newcommand{\CStorageZero}{\ensuremath{\mathcal{Z}{+}}\xspace}

\newcommand{\PStorage}{\ensuremath{\mathcal{P}}\xspace}        

\newcommand{\TOP}[1]{\mathsf{top}_{#1}}     
\newcommand{\push}[1]{\mathsf{push}_{#1}}   
\newcommand{\invpush}[1]{\mathsf{push}^{-1}_{#1}}   
\newcommand{\stay}[1]{\mathsf{stay}_{#1}}   
\newcommand{\pop}[0]{\mathsf{pop}}

\newcommand{\Temptystack}{\ensuremath{\mathit{empty?}}\xspace} 

\newcommand{\N}{\mathbb{N}}

\newcommand{\Aa}{\mathcal{A}}
\newcommand{\Bb}{\mathcal{B}}

\newcommand{\pre}{\mathsf{pre}}
\newcommand{\Ss}{\mathcal{S}}
\newcommand{\succStar}{\mathsf{post}^*}
\newcommand{\preStar}{\mathsf{pre}^*}

\newcommand{\Tt}{\mathcal{T}}
\newcommand{\Mm}{\mathcal{M}}

\newcommand{\id}{\mathsf{id}}

\newcommand{\Reg}{\mathsf{Reg}}

\newcommand{\Encode}{\mathsf{E}}

\newcommand{\set}[2]{%
\ensuremath{\left\{ #1 \mid #2 \right\}}}

\newcommand{\ret}[1]{\mathsf{ret}_{#1}}
\newcommand{\loops}[1]{\mathsf{lp}_{#1}}

\newcommand{\RegularReachProb}{regular backwards reachability problem\xspace}
\newcommand{\ForwardRegularReachProb}{regular forward reachability
  problem\xspace}

\newcommand{\ReachProb}{control state reachability problem\xspace}

\newcommand{\ReachHOCSO}{\texttt{ReachHOCA-}\xspace}
\newcommand{\ReachPDA}{\texttt{ReachPDA}\xspace}
\newcommand{\GeneratePS}{\texttt{GeneratePDA}\xspace}

\newcommand{\overbar}[1]{\mkern 1.5mu\overline{\mkern-1.5mu#1\mkern-1.5mu}\mkern 1.5mu}

\newcommand{\VAL}{\mathsf{VAL}}

\newenvironment{myitemize}{%
\begin{list}{\labelitemi}{%
\setlength{\topsep}{2pt}%
\setlength{\partopsep}{0pt}
\setlength{\itemsep}{0pt}
\setlength{\itemindent}{0ex}
\setlength{\listparindent}{0ex}
\setlength{\leftmargin}{3ex}%
\setlength{\labelwidth}{2ex}
}}
{\end{list}}

\newcommand{\reflong}[1]{%
#1
}

\begin{document}

\title{Reachability in Higher-Order-Counters
\thanks{The second author is supported by the DFG
    research project GELO.
    We both thank
    M.~Boja{\'n}czyk,
    Ch.~Broadbent, and M.~Lohrey
    for helpful discussions and comments.
}}

\author{Alexander Heu\ss ner \inst{1}\and
  Alexander Kartzow \inst{2}}

\institute{ Otto-Friedrich-Universit\"at Bamberg, Germany
\and Universit\"at Leipzig, Germany}

\maketitle

\begin{abstract}
Higher-order counter automata (\HOCS) can be either seen as a
restriction of higher-order pushdown automata (\HOPS) to a unary
stack alphabet, or as an extension of counter automata to higher
levels. We distinguish two principal kinds
of \HOCS: those that can test whether the topmost counter value is
zero and those which cannot.

We show that control-state reachability for level $k$ \HOCS with
$0$-test is complete for \mbox{$(k-2)$}-fold exponential space;
leaving out the $0$-test leads to completeness for
\mbox{$(k-2)$}-fold exponential time.
Restricting \HOCS (without $0$-test) to level $2$, we prove that
global (forward or backward) reachability analysis is
$\PTIME$-complete.
This  enhances the known result for pushdown systems which are
subsumed by level $2$ \HOCS without $0$-test.

We transfer our results to the formal language setting.
Assuming  that
$\PTIME \subsetneq \PSPACE \subsetneq \mathbf{EXPTIME}$,
we apply proof ideas of Engelfriet and conclude that the
hierarchies of languages of \HOPS and of
\HOCS form strictly interleaving hierarchies.
Interestingly, Engelfriet's constructions also allow to conclude immediately
that the hierarchy of collapsible pushdown languages is strict
level-by-level due to the existing complexity results for reachability
on collapsible pushdown graphs. This answers an open question
independently asked  by Parys and by Kobayashi.
\end{abstract}

\section{From Higher-Order Pushdowns to Counters and Back}
Higher-order pushdown automata (\HOPS) --- also known as iterated pushdown
automata ---
were first introduced by Maslov in \cite{Maslov74} and \cite{Maslov76} as an
extension of classical
pushdown automata where the pushdown storage is replaced by a nested pushdown
of pushdowns of\ \dots\ of pushdowns. After being originally studied
as acceptors of languages,
these automata have nowadays obtained renewed interest as
computational model
due to their connection to safe higher-order recursion schemes.
Recent results focus on algorithmic questions
concerning the underlying configuration graphs, e.g., Carayol
and W\"ohrle \cite{cawo03} showed decidability of the monadic
second-order theories of higher-order pushdown graphs due to the
pushdown
graph's connection to the Caucal-hierarchy~\cite{Caucal02}, and Hague and Ong determined the
precise complexity of the global backwards
reachability problem for \HOPS: for level $k$ it is complete for
$\DTIME(\bigcup_{d\in\N}\exp_{k-1}(n^d))$ \cite{HagueOng08}
.\footnote{We define
  $\exp_{0}(n):=n$ and $\exp_{k+1}(n):=\exp(\exp_{k}(n))$ for any natural number $k$.}

In the setting of classical pushdown automata it is well known that restricting the stack alphabet to one single
symbol, i.e., reducing the pushdown storage to a counter,
often makes solving algorithmic problems easier.
For instance, control
state reachability
for pushdown automata is $\PTIME$-complete whereas it is
$\NSPACE(\log(n))$-complete for counter automata.
Then again, results from counter automata raise new insights to the pushdown case by providing algorithmic lower bounds and important subclasses of accepted languages separating different classes of complexity.
In this paper we lift this idea to the higher-order setting by
investigating reachability problems for higher-order counter
automata (\HOCS), i.e., \HOPS over
a one-element stack alphabet.
Analogously to counter automata,
we introduce level $k$ \HOCS in two variants: with or without
$0$-tests. Throughout this paper, we write \HOCSO[k] for the variant
without $0$-tests and \HOCSM[k] for the variant with $0$-tests.
Transferring our results' constructions back to \HOPS will then allow to answer a recent open question~\cite{Parys12,Kobayashi13}.

To our knowledge,
the only existing publication on \HOCS is by Slaats
\cite{Slaats12}.
She proved that
$\HOCSM[(k+1)]$ can simulate level $k$ pushdown
automata (abbreviated $\HOPS[k]$).
In fact, even $\HOCSO[(k+1)]$ simulate $\HOPS[k]$.
Slaats  conjectured that
$L(\HOCSM[k]) \subsetneq L(\HOPS[k])$ where $L(X)$ denotes the
languages accepted by automata of type $X$.
We can confirm this conjecture by combining
the proof ideas of Engelfriet
\cite{Engelfriet91}  with our main result on control-state
reachability for \HOCS in
Theorems \ref{thm:HOCSO-Reachability} and \ref{thm:HOCSM-Reachability}:
control state reachability on $\HOCSM[k]$ is complete for
$\DSPACE(\bigcup_{d\in\N} \exp_{k-2}(n^d))$ and control state reachability on
$\HOCSO[k]$ is complete for
$\DTIME(\bigcup_{d\in\N}\exp_{k-2}(n^d))$. These results
are obtained by adapting a proof strategy relying on reductions to
bounded space storage automata originally stated for \HOPS by
Engelfriet \cite{Engelfriet91}. His main tool are auxiliary
\SPACE{b(n)} $P^k$ automata  where  $P^k$ denotes the storage type of
a $k$-fold nested pushdown (see Section \ref{Chapter_Preliminaries}
for a precise definition). Such a (two-way) automaton has an
additional storage of type $P^k$, and a Turing machine worktape with
space $b(n)$. His main technical result shows  a trade off between
the space bound $b$ and the number of iterated pushdowns $k$. Roughly
speaking, exponentially more space allows to reduce the number of
nestings of pushdowns by one. Similarly, at the cost of another level
of pushdown, one can trade alternation against nondeterminism.
Here, we also restate reachability on $\HOCSM[k]$ as a membership
problem on alternating auxiliary $\mathsf{SPACE}(\exp_{k-3}(n))$
\CStorageZero automata (where \CStorageZero is the new storage type of a
counter with $0$-test).
For our $\DSPACE(\bigcup_{d\in\N}\exp_{k-2}(n^d))$-hardness proof
we provide a reduction of
$\DSPACE(\bigcup_{d\in\N}\exp(\exp_{k-3}(n^d)))$ to alternating auxiliary
$\SPACE{\exp_{k-3}(n)}$ \CStorageZero automata that is
inspired by Jancar  and Sawa's
$\PSPACE$-completeness proof for the non-emptiness of alternating
automata \cite{JancarS07}.  For containment we adapt the
proof of Engelfriet \cite{Engelfriet91} and show that membership for
alternating auxiliary $\SPACE{\exp_{k-3}(n)}$ \CStorageZero automata
can be reduced to alternating reachability on counter automata of
size $\exp_{k-2}(n)$, where $n$ is the size of the original input,
which is known to be in $\DSPACE(\bigcup_{d\in\N}\exp_{k-2}(n^d))$
(cf.~\cite{Goller08}).

For the case of  $\HOCSO[k]$ the hardness follows directly from the
hardness of reachability for level $(k-1)$ pushdown automata and
the fact that the latter can be simulated by $\HOCSO[k]$.
For containment in $\DTIME(\bigcup_{d\in\N}\exp_{k-2}(n^d))$ the
mentioned machinery of
Engelfriet reduces the problem to the case $k=2$.

The proof
that control-state reachability on $\HOCSO[2]$ is in $\PTIME$ is implied by
Theorem\,\ref{thm:RegularReach-2-HOCS} which proves a stronger result: both the
global regular forward and backward reachability problems for $\HOCSO[2]$ are
$\PTIME$-complete. The backward reachability problem asks, given a
regular set $C$ of
configurations, for a (regular) description of all configurations
that allow to reach one in $C$.
This set is typically
denoted as $\mathsf{pre}^*(C)$. Note that there is no canonical way
of defining
a regular set of configurations of $\HOCSO[2]$. We are aware of
at least three possible notions: regularity via $2$-store
automata~\cite{BouajjaniM04},
via sequences of pushdown-operations~\cite{Carayol05}, and via
encoding in regular sets of trees. We stick to the latter,
and use the encoding of configurations as binary
trees introduced in \cite{Kartzow13}: We call a set $C$ of configurations
regular if the set of encodings of configurations
$\set{\Encode(c)}{c\in C}$ is a regular set of trees (where $\Encode$
denotes the encoding function from \cite{Kartzow13}).
Note that the other two notions of regularity are both strictly weaker
(with respect to expressive power)
than the notion of regularity we use here.
Nevertheless, our result does not carry over to these other
notions of regularity as  they
admit more succinct representations of certain sets of configurations.
See  \reflong{Appendix \ref{app:NotionsOfRegularity}} for details.

Besides computing $\pre^*(C)$ in polynomial time our algorithm also
allows to compute the reachable configurations $\succStar(C)$ in polynomial
time. Thus,
 $\HOCSO[2]$ subsumes the well-known
class of pushdown systems~\cite{bouajjani-a-1997-135-a} while still
possessing the same good complexity with
respect to reachability problems.

\section{Formal Model of Higher-Order Counters}
\label{Chapter_Preliminaries}
\subsection{Storage Types and Automata}
An elegant way for defining \HOCS and \HOPS is the use of storage
types and operators on these (following \cite{Engelfriet91}).
For simplicity, we restrict ourselves to what Engelfriet calls
\emph{finitely encoded} storage types.

\begin{definition}
  For $X$ some set, we call a function \mbox{$t:X\to\{true,false\}$} an
  $X$\emph{-test} and a partial function \mbox{$f:X\to X$} an
  $X$\emph{-operation}. \\
  A \emph{storage type} is a tuple
  $\Storage=(\SConfig,T,F,\SConfigInit)$ where
  $\SConfig$ is the set of \Storage-configurations,
  $\SConfigInit\in \SConfig$ the  initial \Storage-configuration,
  $T$ a finite set of $\SConfig$-tests and $F$ a finite set of
  $\SConfig$-operations
  containing  the identity on $\SConfig$, i.e., $\id_{\SConfig}\in F$.
\end{definition}

Let us fix some finite alphabet $\Sigma$ with a distinguished symbol
$\bot\in\Sigma$.
Let $\PStorage_\Sigma = (\SConfig,T,F,\SConfigInit)$ be the \emph{pushdown
  storage type}
where
$\SConfig=\Sigma^+$,
$\SConfigInit= \bot$,
$T=\set{\TOP{\sigma}}{\sigma\in\Sigma}$ with $\TOP{\sigma}(w)=true$ if
$w\in\Sigma^*\sigma$,
and
$F=\set{\push{\sigma}
}{\sigma\in\Sigma}\cup\{\pop,\id\}$
with
$\id=\id_\SConfig$,
$\push{\sigma}(w) = w\sigma$ for all $w\in\SConfig$, and
$\pop(w\sigma)= w$ for all $w\in\Sigma^+$ and $\sigma\in\Sigma$
and $\pop(\sigma)$ undefined for all $\sigma\in\Sigma$.
Hence, $\PStorage_\Sigma$ represents a classical pushdown stack over the
alphabet $\Sigma$. We write $\PStorage$ for $\PStorage_{\{\bot,0,1\}}$.

We define the storage type \emph{counter without $0$-test}
\mbox{$\CStorage = \PStorage_{\{\bot\}}$}, which is  the pushdown
storage over a unary pushdown alphabet.
We define the storage type \emph{counter with $0$-test}
$\CStorageZero$ exactly like $\CStorage$ but  we
add the test $\Temptystack$ to the set of tests where
$\Temptystack(\Sconfig)=true$ if $\Sconfig=\bot$ (the plus in
$\CStorageZero$ stands for
``with $0$-test''). In other words,
$\Temptystack$ returns false iff the operation $\pop$ is
applicable.
\begin{definition}
  For a storage type $\Storage=(\SConfig,T,F,\SConfigInit)$ we
  define an \emph{$\Storage$ automaton} as a tuple $\Aa=(Q, q_0, q_f, \Delta)$
  where as usual $Q$ is a finite set of states with initial state $q_0$
  and final state $q_f$ and $\Delta$ is the transition
  relation. The difference to a usual automaton is the definition of
  $\Delta$ by $\Delta= Q\times \{ true,false\}^T \times Q \times F$.
\end{definition}
For $q\in Q$ and $\Sconfig\in\SConfig$, a transition $\delta=(q, R, p, f)$ is
applicable to the \emph{configuration} $(q,\Sconfig)$
if $f(\Sconfig)$ is defined
and if for each test
$t\in T$ we have $R(t) = t(\Sconfig)$, i.e., the result of the
storage-tests on the storage configuration $x$ agree with the test
results required by the transition $\delta$. If $\delta$ is applicable,
application of $\delta$ leads to the configuration
$(p,f(\Sconfig))$. The notions of a run, the accepted language, etc.~are
now all defined as expected.

\paragraph{The Pushdown Operator}
We also consider $\PStorage_\Sigma$ as an
operator on other storage types as follows.
Given a storage type
$\Storage=(\SConfig,T,F,\SConfigInit)$
let the
storage type \emph{pushdown of \Storage}
be
$\PStorage_\Sigma(\Storage)=(\SConfig',T',F',\SConfigInit')$ where
$\SConfig'=(\Sigma\times \SConfig)^+$,
$\SConfigInit'= (\bot,\SConfigInit)$,
$T'=\set{\TOP{\sigma}}{\sigma\in\Sigma} \cup
    \set{test(t)}{t\in T}$,
$F'=\set{\push{\gamma,f}}{\gamma\in\Sigma, f\in F}\cup
    \set{\stay{f}}{f\in F} \cup
    \{pop\}$, and where
for all
$\Sconfig'=\beta(\sigma,\Sconfig)$, $\beta\in(\Sigma\times \SConfig)^*$,
$\sigma\in\Sigma$, $\Sconfig\in \SConfig$ it holds that
\begin{myitemize}
\item $\TOP{\tau}(\Sconfig') = (\tau=\sigma)$,
\item $test(t)(\Sconfig')=t(\Sconfig)$,
\item $\push{\tau,f}(\Sconfig')= \beta(\sigma,\Sconfig)(\tau,f(\Sconfig))$\\
  if $f$ is defined on $\Sconfig$ (and undefined otherwise),
\item $\stay{f}(\Sconfig')= \beta(\sigma,f(\Sconfig))$\\
  if $f$ is defined on $x$ (and undefined otherwise), and
\item $\pop(\Sconfig')=\beta$ if $\beta$ is nonempty (and undefined
  otherwise).
\end{myitemize}
\noindent Note that $\stay{\id_\SConfig}=\id_{\SConfig'}$ whence $F'$ contains
the identity.
As for storages, we define the
operator $\PStorage$ to be the operator $\PStorage_{\{\bot,0,1\}}$.

\subsection{\HOPS, \HOCS, and their Reachability Problems}
We can define the iterative application of the operator $\PStorage$ on
some storage $\Storage$ as
follows:
let $\PStorage^0(\Storage)=\Storage$ and
$\PStorage^{k+1}(\Storage)={\PStorage}(\PStorage^k(\Storage))$.
A \emph{level $k$ higher-order pushdown automaton}  is a
$\PStorage^{k-1}(\PStorage)$ automaton.
We abbreviate the class of all these automata with \HOPS[k].
A \emph{level $k$ higher-order counter automaton with zero-test}
is a $\PStorage^{k-1}(\CStorageZero)$ automaton and
\HOCSM[k] denotes the corresponding class.\footnote{A priori
 our definition of
$\HOCSM[k]$ results in a stronger automaton model than that
used by Slaats. In fact, both models are equivalent (cf.~\reflong{Appendix
\ref{app:Storage_Equivs}}).}.
Similarly, \HOCSO[k] denotes the class of
\emph{level $k$ higher-order counter
automata without zero-test} which is the class of
$\PStorage^{k-1}(\CStorage)$ automata. Obviously,
for any level $k$ it holds that
$L(\HOCSO[k])\subseteq L(\HOCSM[k])\subseteq{L(\HOPS[k])}$
where $L(X)$ denotes the
languages accepted by automata of type $X$.

We next define the
reachability problems which
we study in this paper.
\begin{definition}
  Given an \Storage automaton and one of its control
  states $q\in Q$, then the \emph{\ReachProb} asks whether there
  is a configuration $(q,\Sconfig)$ that  is reachable from
  $(q_0,\SConfigInit)$ where $\Sconfig\in \SConfig$ is an
  arbitrary \Storage-configuration.
\end{definition}

Assuming a notion of regularity for sets of \Storage configurations
(and hence for sets of configurations of \Storage automata),
 we can also define a global variant of the \ReachProb.

\begin{definition}
  Given an \Storage automaton $\Aa$ and a regular
  set of configurations $C$, the \emph{\RegularReachProb} demands a
  description of the set of configurations from which there is a path
  to some configuration $c\in C$.
\end{definition}

Analogously, the \ForwardRegularReachProb asks for a description of
the set of configurations reachable from a given regular set $C$.
In the following section, we  consider the regular backwards (and
forwards) reachability problem for the
class of $\HOCSO[2]$ only.

\section{Regular Reachability for \HOCSO[2]}
\label{sec:RegularReachability}

The goal of this section is to prove the following theorem extending a known result on regular reachability on
pushdown systems to \HOCSO[2]:

\begin{theorem}
\label{thm:RegularReach-2-HOCS}
  Reg. backwards/forwards reachability on \HOCSO[2] is
  $\PTIME$-complete.
\end{theorem}

\subsection{Returns, Loops, and Control State Reachability}
\label{sec:Returns}

Proving Theorem\,\ref{thm:RegularReach-2-HOCS}
is based on  the ``returns-\&-loops''
construction for \HOPS[2] of~\cite{Kartzow13}. As a first step,
we  consider the simpler case of control-state reachability:

\begin{proposition}
   \label{prop:ReachHOCS-0}
  Control state reachability for \HOCSO[2]
  is $\PTIME$-complete.
\end{proposition}

In \cite{Kartzow13} it has been shown that certain runs,
so-called
\emph{loops} and \emph{returns}, are
the building blocks of any run of a \HOPS[2] in the sense that
solving a reachability problem amounts to deciding whether
certain loops and returns exist. Here, we analyse these notions more
precisely in the context of
\HOCSO[2] in order to derive
a polynomial control state reachability algorithm.
Using this algorithm we can then also solve the \RegularReachProb
efficiently.

{For this section, we fix a ${\PStorage}(\CStorage)$-automaton
  $\Aa=(Q, q_0, F, \Delta)$}.
Recall that the ${\PStorage}(\CStorage)$-configurations of $\Aa$ are
elements of $(\Sigma\times \{\bot\}^+)^+$. We identify $\bot^{m+1}$
with the natural number $m$ and the set of storage configurations with
$(\Sigma\times \N)^+$.

\begin{definition}
  Let $s\in (\Sigma\times\N)^+$, $t \in \Sigma\times \N$
  and $q,q'\in Q$ be states of $\Aa$.
  A \emph{return} of $\Aa$ from $(q,st)$ to $(q',s)$ is  a run $r$ from
  $(q,st)$ to $(q',s)$ such that
  except for the final configuration no
  configuration of $r$ is in $Q\times \{s\}$.

  Let $s\in (\Sigma\times\N)^*$, $t \in \Sigma\times \N$.
  A \emph{loop} of $\Aa$ from $(q,st)$ to $(q',st)$ is  a run $r$ from
  $(q,st)$ to $(q',st)$ such that no configuration of $r$ is in
  $Q\times\{s\}$.
\end{definition}

One of the underlying reasons why control state reachability for
pushdown systems can be efficiently solved is the fact that it is
always possible to reach a certain state without increasing the
pushdown by more than polynomially many elements.
In the following, we prove an analogue of this fact for
${\PStorage}(\CStorage)$.  For a given configuration, if there is a
return or loop starting in this configuration, then this return or
loop can be realised without increasing the (level 2) pushdown
more than polynomially.
This is due to the
monotonic behaviour of $\CStorage$: given a $\CStorage$ configuration
$\Sconfig$, if we can apply a sequence $\varphi$ of
transitions to $\Sconfig$ then we can apply $\varphi$
to all bigger configurations, i.e., to any configuration of the form
$\push{\bot}^n(\Sconfig)$. Note that this depends on the fact that
$\CStorage$ contains only trivial tests (the test $\TOP{\bot}$
always returns true). In contrast, for $\CStorageZero$, if $\varphi$
applies a couple of $\pop$ operations and then tests for zero and
performs a transition, then this is not
applicable to a bigger counter because the $0$-test would now fail.

For a ${\PStorage}(\CStorage)$
configuration $\Sconfig=(\sigma_1,
n_1)(\sigma_2,n_2)\dots(\sigma_m,n_m)$, let
$\lvert \Sconfig\rvert = m$ be its \emph{height}.
Let $r$ be some run starting in $(q, \Sconfig)$ for some $q\in
Q$. The run $r$
\emph{increases the height by at most $k$} if
$\lvert \Sconfig'\rvert \leq \lvert \Sconfig\rvert + k$
for all configurations  $(q',\Sconfig')$ of $r$.

\begin{definition}
  Let $s\in (\{\bot\}\times \N)^+$.
  We write $\ret{k}(s)$ and $\loops{k}(s)$, resp., for the
  set
  of pairs of initial and final
  control states of returns or loops starting in  $s$ and
  increasing the height by at most $k$.
  We write $\ret{\infty}(s)$ and $\loops{\infty}(s)$,resp., for
  the union of all $\ret{k}(sw)$ or  $\loops{k}(s)$.
\end{definition}

The existence of a return (or loop) starting in
$sw$ \mbox{(or $s'w$)} (with $s\in (\{\bot\}\times \N)^+, s'\in
(\{\bot\}\times \N)^*$ and $w\in \{\bot\}\times \N$)
does
not depend on the concrete choice of $s$ or $s'$. Thus, we also write
$\ret{k}(w)$ for $\ret{k}(sw)$ and $\loops{k}(w)$ for
$\loops{k}(s'w)$.

By induction on the length of a run, we first prove that
${\PStorage}(\CStorage)$ is \emph{monotone} in the following sense:
let $s\in(\Sigma\times\N)^*, t=(\sigma,n)\in\Sigma\times\N$,
$q,q'\in Q$ and $r$ a run starting in $(q, st)$ and ending in state
$q'$. If the topmost counter of each configuration of $r$ is at least
$m$, then for each $n'\geq n-m$ there is a run $r'$ starting in
$(q, s(\sigma,n'))$ and performing exactly the same transitions as $r$.
In particular, for all $k\in\N\cup\{\infty\}$, $\sigma\in\Sigma$ and
$m_1\leq m_2\in\N$,
  $\ret{k}((\sigma,m_1))\subseteq \ret{k}((\sigma,m_2))$ and
  $\loops{k}((\sigma,m_1))\subseteq \loops{k}((\sigma,m_2))$.

We next show that the sequence $(\ret{k}((\sigma,m)))_{m\in\N}$
stabilises at $m=\lvert \Sigma\rvert \lvert Q \rvert^2$. From this we
conclude that $\ret{\infty}=\ret{\lvert \Sigma\rvert^2 \lvert Q
  \rvert^4}$, i.e., in order to realise a return with arbitrary
fixed initial and final configuration, we do not have to
increase the height by more than
$\lvert \Sigma\rvert^2 \lvert Q \rvert^4$ (if there is such a return
at all).

\begin{lemma}\label{lem:BoundCounterForKReturns}
  For $k\in\N\cup\{\infty\}$, $\sigma\in\Sigma$,
  $m\geq\lvert \Sigma\rvert \lvert Q\rvert^2$, and
  $m'\geq  2 \cdot \lvert \Sigma\rvert \lvert Q\rvert^2$,
  we have
  $\ret{k}((\sigma,m))=\ret{k}((\sigma,\lvert \Sigma\rvert \lvert Q
  \rvert^2))$ and
  $\loops{\infty}((\sigma,m'))=\loops{\infty}((\sigma,2 \cdot \lvert
  \Sigma\rvert \lvert Q  \rvert^2))$.
\end{lemma}
The proof uses the fact that we can find an $m'\leq \lvert\Sigma\rvert \lvert
Q\rvert^2$ with $\ret{k}((\sigma,m'))= \ret{k}(\sigma,m'+1)$ for all
$\sigma$ by the pigeonhole-principle. Using monotonicity of
$\PStorage(\CStorage)$ we conclude that
$\ret{k}(\sigma,m')=\ret{k}(\sigma,m)$ for all $m\geq m'$.
A similar application of the pigeonhole-principle shows that there is
a $k\leq \lvert \Sigma\rvert^2 \cdot \lvert Q\rvert^4$ such that
$\ret{k}((\sigma,i))=\ret{k+1}((\sigma,i))$ for all $\sigma$ and all
$i \leq \lvert \Sigma\rvert \lvert Q \rvert^2$ (or equivalently for
all $i\in\N$). By induction on $k'\geq k$ we show that
$\ret{k'}=\ret{k}$ because
any subreturn that
increases the height by $k+1$ can be replaced by a subreturn that only
increases the height by $k$. Thus, we obtain the following lemma.

\begin{lemma} \label{lem:ZerolessQtoFourReturnsareEnough}
  For all $i\in\N$ and $\sigma\in\Sigma$, we have
  $\ret{\infty}((\sigma,i))=\ret{\lvert \Sigma\rvert^2 \cdot \lvert
    Q\rvert^4}((\sigma,i))$ and
  $\loops{\infty}=\loops{\lvert \Sigma\rvert^2\lvert Q \rvert^4+1}$.
\end{lemma}

\begin{figure}[t]
  \scriptsize
  \texttt{{\GeneratePS}($\Aa,A$\texttt{):}}\\
  {\bf Input:} \HOCSO[2] $\Aa=(Q,q_0,\Delta)$ over $\Sigma$,
 matrix $A=(a_{\sigma,p,q})_{(\sigma,p,q)\in \Sigma\times Q^2}$
     over $\N\cup \{\infty\}$\\
     {\bf Output:} \HOPS[1] $\Aa'$ simulating $\Aa$
  \begin{lstlisting}
$k_0$ := $\lvert \Sigma\rvert^2 \cdot \lvert Q\rvert^4$; $h_0$ :=$\lvert \Sigma \rvert \cdot \lvert Q^2\rvert$; $\Delta'$ := $\emptyset$
foreach $\delta\in\Delta$:
   if $\delta$ == $(q, (\sigma,\bot), \stay{\pop}, p)$:
      foreach $i$ in $\{0, \dots, h_0\}$: $\Delta'$:=$\Delta'\cup\{ ((q,\sigma),\bot_i, \pop, (p,\sigma)),((q,\sigma), ,\bot_\infty, \pop, (p,\sigma))\}$
   elseif $\delta$==$(q,(\sigma,\bot),\stay{\push{\bot}},p)$
      $\Delta'$:=$\Delta' \cup  \{ ((q,\sigma),\bot_{\infty},\push{\bot_\infty},  (p,\sigma))\} \cup  \{ ((q,\sigma),\bot_{h_0}, \push{\bot_\infty},  (p,\sigma))\}$
      foreach $i$ in $\{0, \dots, h_0-1\}$}: $\Delta':=\Delta' \cup  \{ ((q,\sigma),\bot_i, \push{\bot_{i+1}},  (p,\sigma))\}$
   elseif $\delta$==$(q, (\sigma,\bot), \push{\tau,\id}, p)$
      foreach $r$ in $Q$ such that $a_{\tau,p,r}\neq\infty$:
          foreach $i$ in $\{a_{\tau,p,r}, a_{\tau,p,r}+1, \dots, h_0\}\cup\{\infty\}$: $\Delta':=\Delta'\cup\{((q,\sigma), \bot_i, \id, (r,\sigma))$
$\Aa'$:=$(Q\times\Sigma,  (q_0,\bot), \Delta')$
return $\Aa'$
  \end{lstlisting}
  \caption{\HOCSO[2] to \HOPS[1] Reduction Algorithm}\label{algo:GeneratePS}
\end{figure}

We now can prove that control-state reachability on \HOCSO[2] is
$\PTIME$-complete.

\begin{proof}[of Proposition \ref{prop:ReachHOCS-0}]
  Since $\HOCSO[2]$  can trivially simulate pushdown automata,
  hardness follows from the analogous hardness result for pushdown automata.
  Containment in $\PTIME$  uses the
  following ideas:
  \begin{enumerate}
  \item We assume that
    the input $(\Aa, q)$ satisfies that
    $q$ is reachable in
    $\Aa$ iff $(q, (\bot,0))$ is reachable
    and that $\Aa$ only uses instructions of the forms
    $\pop{}$, $\push{\sigma,\id}$, and $\stay{f}$. Given any
    $\HOCSO[2]$ $\Aa'$ and a state $q$, it is straightforward to
    construct (in polynomial time) a $\HOCSO[2]$ $\Aa$ that satisfies
    this condition such that $q$ is reachable in $\Aa'$ iff
    it is reachable in $\Aa$.
  \item
    Recall that $\ret{\infty}(w)=\ret{k_0}(w)$ for
    $k_0=\lvert \Sigma\rvert^2 \cdot
    \lvert Q\rvert^4$ and for all $w\in \Sigma\times \N$.
    Set $h_0=\lvert \Sigma \rvert \cdot \lvert Q^2\rvert$.
    We want to compute a table $(a_{\sigma,p,q})_{{\sigma,p,q}\in
      \Sigma\times Q^2}$ with
    values in
    $\{\infty, 0, 1, 2, \dots, h_0\}$   such that
    $a_{\sigma,p,q}=\min\{i \mid  (p,q)\in\ret{k_0}((\sigma,i))\}$
    (where we set $\min\{\emptyset\}=\infty$).
    Due to Lemmas \ref{lem:BoundCounterForKReturns} and
    \ref{lem:ZerolessQtoFourReturnsareEnough} such a table represents
    $\ret{\infty}$  in the sense that
    $(p,q)\in\ret{\infty}((\sigma,i))$ iff $i\geq a_{\sigma,p,q}$.
  \item With the help of the table
    $(a_{\sigma,p,q})_{(\sigma,p,q)\in \Sigma\times Q^2}$ we
    compute in polynomial time a $\PStorage$ automaton
    $\Aa_{\infty}$ which executes
    the same level $1$ transitions as $\Aa$ and simulates
    loops of $\Aa$ in the following
    sense: if there is a loop of $\Aa$ starting in $(q,(\sigma,i))$
    performing first a $\push{\tau,\id}$ operation  and then
    performing a  return with final  state $p$, we allow $\Aa'$ to
    perform an $\id$-transition from $(q,(\sigma,i))$ to
    $(p,(\sigma,i))$.  This new system basically keeps track of the
    height of the
    pushdown up to $h_0$ by
    using a pushdown alphabet $\{\bot_0, \dots, \bot_{h_0}, \bot_\infty\}$
    where the topmost symbol of the pushdown is $\bot_i$ iff the
    height of the pushdown is $i$ (where $\infty$ stands for values
    above $h_0$). After this change of pushdown alphabet, the
    additional $\id$-transitions   are easily computable from
    the table $(a_{\sigma,p,q})_{(\sigma,p,q)\in \Sigma\times Q^2}$.
    The resulting system has size
    $O( h_0^2 \cdot
    (\lvert \Ss \rvert +1))$, i.e., is polynomial in the original
    system $\Aa$.
  \item Using
    \cite{bouajjani-a-1997-135-a},
    check for reachability of $q$ in
    the pushdown automaton  $\Aa_{\infty}$.
  \end{enumerate}
  In fact, for step $2$ we already use a variant of
  steps $3$ and $4$: we  compute
  $\ret{\infty}= \ret{\lvert \Sigma\rvert^2\lvert Q \rvert^4}$
  by induction starting with $\ret{0}$.
  If we remove all level $2$ operations from $\Aa$ and store the
  topmost level $2$ stack-symbol in the control state we obtain a
  pushdown automaton $\Bb$ such that $(q,q')\in\ret{0}(\sigma,k)$
  (w.r.t.~$\Aa$) iff there is a transition $(p,(\sigma,\bot),\pop,q')$
  of $\Aa$ and the control state $(p,\sigma)$ is reachable from
  $((p,\sigma),k)$ in $\Bb$. Thus, the results of polynomially many
  reachability queries for $\Bb$ determine the table for $\ret{0}$.
  Similarly, we can use the table of $\ret{i}$ to compute the table of
  $\ret{i+1}$ as follows. A return extending the height of the
  pushdown by $i+1$ decomposes into parts that do not increase the
  height at all and parts that perform a $\push{\tau,\id}$ followed by
  a return increasing the height by at most $i$. Using the table for
  $\ret{i}$ we can easily enrich $\Bb$ by $\id$-transitions that
  simulate such push operations followed by returns increasing the
  height by at most $i$. Again,  determining whether
  $(q,q')\in\ret{i+1}(\sigma,k)$ reduces to one reachability query
  on this enriched $\Bb$ for each pop-transition of $\Aa$.

  With these ideas in mind, it is straightforward to check that
  algorithm \ReachHOCSO
  in Figure\,\ref{algo:Reach-2HOCSO} (using  algorithm \GeneratePS
  of Figure\,\ref{algo:GeneratePS} as subroutine for step $3$)
  solves the reachability problem for
  \HOCSO[2] (of the form described in step $1$) in polynomial time.
  In this algorithm,
  \ReachPDA($\Aa', c,q$) refers to the classical
  polynomial time algorithm that
  determines  whether in the (level $1$) pushdown automaton $\Aa'$
  state $q$ is reachable when  starting in configuration $c$; a
  transition $(q,(\sigma,\tau),f,p)$ refers to a transition from state
  $q$ to state $p$ applying operation $f$ that is executable if
  the (level 2) test $\TOP{\sigma}$ and the (level 1) test
  $test(\TOP{\tau})$ both succeed.
  \qed
\end{proof}

\begin{figure}[t!]
  \scriptsize
  \texttt{{\ReachHOCSO}(}$\Aa,q_f$\texttt{):}\\
  {\bf Input:} \HOCSO[2] $\Aa=(Q,q_0,\Delta)$ over $\Sigma$, $q_f\in
  Q$\\
  {\bf Output:} whether $q_f$ is reachable in $\Aa$
  \begin{lstlisting}
$k_0$ := $\lvert \Sigma\rvert^2 \cdot \lvert Q\rvert^4$; $h_0$ :=$\lvert \Sigma \rvert \cdot \lvert Q^2\rvert$;
foreach $(\sigma,p,q)$ in $\Sigma\times Q^2$: $a_{\sigma,p,q}:=\infty$
for $k=1, 2, \dots, k_0$:
   $\Aa_k:=$ GeneratePDA($\Aa, (a_{\sigma,p,q})_{(\sigma,p,q)\in \Sigma\times Q^2}$)
   foreach $(r,(\tau,\bot),\pop,q)$ in $\Delta$ and $(\sigma,p)$ in $\Sigma\times Q$:
      for $i=h_0, h_0-1, \dots, 1,0$:
      if ReachPDA($\Aa_k, ((p, \sigma), i), (r,\tau)$): $a'_{\sigma,p,q}:=i$
   foreach $(\sigma,p,q)$ in $\Sigma\times Q^2$: $a_{\sigma,p,q}$:=$a'_{\sigma,p,q}$
$\Aa_\infty$:=GeneratePDA($\Aa, (a_{\sigma,p,q})_{(\sigma,p,q)\in \Sigma\times Q^2}$)
if Reach($\Aa_\infty, ((q_0,\bot), 0), (q_f,\bot)$): return true else return false
\end{lstlisting}
\caption{Reachability on $2$-\HOCSO Algorithm 2
  \label{algo:Reach-2HOCSO}}
\end{figure}

\subsection{Regular Reachability}
\label{sec:RegReach}

In order to define regular sets of configurations, we
recall the  encoding of
\HOPS[2] configurations as trees
from \cite{Kartzow13}.
Let  $p=(\sigma_1,v_1)(\sigma_2,v_2)\dots(\sigma_m,v_m)\in
\PStorage(\CStorage)$.
If $v_1=0$, we set $p_l=\emptyset$ and $p_r=
(\sigma_2,v_2)\dots(\sigma_m,v_m)$.
Otherwise, there is a maximal $1\leq j \leq m$ such that $v_1, \dots,
v_j\geq 1$ and we set $p_l=(\sigma_1,v_1-1)\dots(\sigma_j,v_j-1)$ and
$p_r=(\sigma_{j+1},v_{j+1})\dots(\sigma_m,v_m)$ if $j<m$ and
$p_r=\emptyset$ if $j=m$.
The \emph{tree-encoding} $\Encode$ of $p$ is given as follows:\\[-1.5ex]
\begin{minipage}{0.7\linewidth}
  \begin{align*}
    \Encode(p)=
    \begin{cases}
      \emptyset &\text{if } p=\emptyset\\
      \bot(\sigma_1,\Encode(p_r)) & \text{if }p=(\sigma_1,0) p_r\\
      \bot(\Encode(p_l), \Encode(p_r)) & \text{otherwise},
    \end{cases}
  \end{align*}
\end{minipage}
\begin{minipage}[r]{0.25\linewidth}
  \begin{align*}
    \scalebox{0.5}{
      \begin{tikzpicture}
        [
        each edge/.style=
        {draw, shorten <= 1, shorten >= 1},
        gu/.style={grow=right, level distance=20},
        gr/.style={grow=up, level distance=20}
        ]
        \node (root)  {$q$}
        child[gr]{
          node {$\bot$}
          child[gr]{
            node {$\bot$}
          child[gr]{
            node {$\bot$}
            child[gr]{
              node {$a$}
            }
            child[gu]{
              node{$\bot$}
              child[gr]{
                node{$a$}
              }
            }
          }
        }
        child[gu, level distance=40]{
          node {$\bot$}
          child[gr]{
            node{$a$}
          }
          child[gu]{
            node {$\bot$}
            child[gr]{
              node{$\bot$}
              child[gr]{
                node{$b$}
              }
            }
          }
        }
      }
      ;
    \end{tikzpicture}
  }
\end{align*}
\end{minipage}\\
where  $\bot(t_1,t_2)$ is the tree with root $\bot$ whose left
subtree is $t_1$ and whose right subtree is $t_2$.
For a configuration $c=(q,p)$ we define
$\Encode(c)$ to be the tree $q(\Encode(p),\emptyset)$. The picture
beside the definition of $\Encode$ shows the encoding of the configuration
$\left(q,(a,2)(a,2)(a,0)(b,1)\right)$.
Note that for each element $(\sigma,i)$ of $p$, there
is a path to a leaf $l$ which is labelled by $\sigma$ such that the path
to $l$ contains $i+2$ left successors. Moreover, the inorder traversal
of the tree induces an order of the leaves which corresponds to the
left-to-right order of the elements of $p$.
We call a set $C$ of configurations \emph{regular} if the set
$\{\Encode(c) \mid c\in C\}$ is a regular set of trees.

$\Encode$ turns
the reachability predicate on $\HOCSO[2]$ into
a tree-automatic relation \cite{Kartzow13}, i.e., for a given
$\HOCSO[2]$ $\Aa$, there is a
tree-automaton $\Tt_\Aa$ accepting the convolution of $\Encode(c_1)$
and $\Encode(c_2)$ for $\HOCSO[2]$ configurations $c_1$ and $c_2$ iff
there is a run of $\Aa$ from $c_1$ to $c_2$.
This allows to solve the regular backwards reachability problem as
follows. On input a $\HOCSO[2]$ and a tree automaton $\Tt$ recognising
a regular set $C$ of configurations, we first compute the tree-automaton
$\Tt_\Aa$. Then using a simple product construction of $\Tt_\Aa$ and
$\Tt$ and projection, we obtain an automaton $\Tt_{\mathsf{pre}}$
which accepts $\preStar(C)=\{\Encode(c) \mid \exists c'\in C \text{
  and a run from } c \text{ to } c'\}$.
The key issue for the complexity of this construction is the
computation of $\Tt_\Aa$ from $\Aa$.
The explicit construction of $\Tt_\Aa$ in
\cite{Kartzow13} involves an
exponential blow-up. In this construction
the blow-up is only caused by a part of $\Tt_\Aa$ that computes
$\ret{\infty}(\sigma,m)$ for each $\sigma\in\Sigma$ on input
a path whose labels form the word $\bot^m$.
Thus, we can exhibit the following consequence.

\begin{corollary}[\cite{Kartzow13}]\label{thm:KartzohPHDMain}
  Given a $2$-\HOCSO $\Aa$ with state set $Q$, we can compute
  the tree automaton $\Tt_\Aa$ in $\PTIME$, if we can compute from
  $\Aa$ in $\PTIME$ a deterministic word automaton $\Tt'$ with state
  set  $Q'\subseteq \prod_{\sigma\in\Sigma} (2^{Q\times Q})^2$ such
  that for all $m\in\N$ the state of $\Tt'$ on input $\bot^m$ is
  $\left(\ret{\infty}(\sigma, m),
  \loops{\infty}(\sigma,m)\right)_{\sigma\in\Sigma}$.
\end{corollary}

Thus, the following lemma completes the proof of
Theorem
\ref{thm:RegularReach-2-HOCS}.

\begin{lemma}\label{lem:MainLemComputationLpRet}
  Let $\Aa$ be a \HOCSO[2] with state set $Q$.
  We can compute in polynomial time a deterministic
  finite word automaton $\Aa'$ with state set $Q'$ of size at most
  $2\cdot (\lvert \Sigma \rvert \cdot \lvert Q\rvert^2+1)$ such that
  $\Aa'$ is in state
  $\left(\ret{\infty}((\sigma,n)),
    \loops{\infty}((\sigma,n))\right)_{\sigma\in\Sigma}$
  after reading $\bot^n$   for every $n\in\N$.
\end{lemma}
\begin{proof}
  Let $n_0=2 \cdot \lvert
  \Sigma\rvert \cdot  \lvert Q \rvert^2$.
  Recall algorithm \ReachHOCSO of Figure\,\ref{algo:Reach-2HOCSO}. In
  this  polynomial time
  algorithm we computed a matrix
  $A=(a_{\sigma,p,q})_{(\sigma,p,q)\in\Sigma \times Q^2}$ representing
  $\ret{\infty}$ and a pushdown automaton $\Aa_\infty$ (of level $1$)
  simulating $\Aa$ in the sense that $\Aa_\infty$ reaches a
  configuration $((q,\sigma) p)$ for a pushdown $p$ of height $n$ if and
  only if $\Aa$ reaches $(q, (\sigma, n))$.
  It is sufficient to describe a
  polynomial time algorithm that computes
  $M_i= \left(\ret{\infty}((\sigma,n)),
  \loops{\infty}((\sigma,n))\right)_{\sigma\in\Sigma}$ for all $n\leq n_0$.
  $\Aa'$ is then the automaton with state set $\{M_i\mid i\leq n_0\}$,
  transitions from $M_i$
  to $M_{i+1}$ for each $i< n_0$ and a transition from $M_{n_0}$ to
  $M_{n_0}$. The correctness of this construction follows from Lemma
  \ref{lem:BoundCounterForKReturns}.

  Let us now describe how to compute $M_i$ in polynomial time.
  Since $\Aa_{\infty}$ simulates $\Aa$ correctly, there is a loop from
  $(q, (\sigma,i))$ to $(q', (\sigma,i))$ of $\Aa$ if and only if there
  is a run of $\Aa_{\infty}$ from $((q,\sigma), p_i)$ to
  $((q',\sigma), p_i)$ for $p_i=\bot_0\bot_1\dots\bot_i$ (where we
  identify $\bot_j$ with $\bot_\infty$ for all $j> h_0$). Thus, we can compute
  the loop part of  $M_i$ by $n_0$ many calls to an algorithm for
  reachability on pushdown systems.
  Note that $(p,q)\in\ret{\infty}((\sigma,i))$ with respect to $\Aa$
  if there is a state $r$ and some $\tau\in\Sigma$ such that
  $(r, \tau, \pop{}, q)$ is a transition of $\Aa$ and
  $(r,\tau)$ is reachable in $\Aa_{\infty}$ from $((q,\sigma),i)$.
  Thus, with a loop over all transitions of $\Aa$ we reduce the
  computation of the returns component of $M_i$ to polynomially many
  control state reachability problems on a pushdown system.\qed
\end{proof}

\section{Reachability for \HOCSO[k] and \HOCSM[k]}
\label{sec:ReachGeneral}
Using slight adaptations of Engelfriet's seminal paper
\cite{Engelfriet91}, we can lift the result on reachability for
\HOCSO[2] to reachability for \HOCSO[k]
(cf.~\reflong{Appendix \ref{app:ReachGeneral}}).

\begin{theorem}\label{thm:HOCSO-Reachability}
  For $k\geq 2$,
  the control state reachability problem for \HOCSO[k] is
  complete for $\DTIME(\bigcup_{d\in\N}\exp_{k-2}(n^d))$.
  For $k\geq 1$,
  the alternating control state reachability problem for \HOCSO[k]
  is complete for
  $\DTIME(\bigcup_{d\in\N}\exp_{k-1}(n^d))$.
\end{theorem}
Hardness follows from the hardness of control state reachability for
\HOPS[(k-1)] \cite{Engelfriet91} and the
trivial fact that the storage type $\PStorage^{k-1}$ of \HOPS[(k-1)]
can be trivially simulated by the storage type
$\PStorage^{k-1}(\CStorage)$ of  \HOCSO[k].
Containment for the first claim is proved by induction on $k$ (the
base case $k=2$ has been proved in the previous section). For $k\geq
3$, we
use  Lemma 7.11,  Theorems 2.2 and 2.4 from \cite{Engelfriet91}
and reduce reachability of $\HOCSO[k]$ to reachability on
(exponentially bigger) $\HOCSO[(k-1)]$.
For the second claim, we adapt Engelfriet's  Lemma 7.11 to a version
for the setting of alternating automata (instead of
nondeterministic automata)  and use his Theorems 2.2. and 2.4 in order
to show equivalence (up to logspace reductions) of alternating
reachability for $\HOCSO[(k-1)]$ and reachability for $\HOCSO[k]$.

We can also reduce reachability for
 $\HOCSM[k]$ to  reachability for $(k-1)$-fold exponentially
bigger $\HOCSM[1]$. Completeness for $\NSPACE(\log(n))$
of  reachability for $\HOCSM[1]$ (cf.~\cite{Goller08})
yields the  upper bounds for  reachability for
$\HOCSM[k]$. The corresponding lower bounds follow by
applications of Engelfriet's theorems and
an adaptation of the $\PSPACE$-hardness proof for emptiness
of alternating finite automata by Jancar and Sawa
\cite{JancarS07}.

\begin{theorem}\label{thm:HOCSM-Reachability}
  For $k\geq 2$, (alternating) control state reachability for \HOCSM[k] is
  complete for ($\DSPACE(\bigcup_{d\in\N} \exp_{k-1}(n^d))$)
  $\DSPACE(\bigcup_{d\in\N} \exp_{k-2}(n^d))$.
\end{theorem}

\section{Back to HOPS: Applications to Languages}
\label{sec:Languages}
Engelfriet \cite{Engelfriet91} also  discovered a close connection
between the complexity
of the control state reachability problem for a class of automata and
the class of languages recognised by this class.
We restate a slight extension (cf.~\reflong{Appendix\ref{app:Languages}})
of these results and use them to
confirm Slaat's conjecture from \cite{Slaats12}.
\begin{proposition}\label{prop:EngelfrietLanguageReachConnection}
  Let $\Storage_1$ and $\Storage_2$ be storage types and $C_1,C_2$ complexity
  classes such that $C_1\subsetneq C_2$. If control state reachability for
  nondeterministic $\Storage_i$ automata is complete for $C_i$, then
  there is a deterministic $\Storage_2$ automaton accepting some
  language $L$ such that no nondeterministic $\Storage_1$-automaton
  accepts $L$.
\end{proposition}

In fact, Engelfriet's proof can be used to derive a separating
language.
For a storage type $\Storage=(\SConfig,T,F,\SConfigInit)$, we
define the \emph{language of valid storage sequences} $\VAL(\Storage)$
as follows.
For each test $t\in T$ and $r\in\{true,false\}$ we set
$t_r:=\id{\restriction}_{\{\Sconfig\in\SConfig\mid t(\Sconfig)=r
  \}}$ and set
$\Sigma= F\cup \{ t_r \mid t\in T, r\in\{true,false\}\}$.
For $s\in \Sigma^*$ such that $s=a_1\dots
a_n$, and
$\Sconfig\in\SConfig$ we write $s(\Sconfig)$ for
$a_n(a_{n-1}(\dots a_1(\Sconfig)\dots))$. We define
  $\VAL(\Storage) =
  \left\{ s\in \Sigma^* \left|
      s(\SConfigInit) \text{ is defined}
    \right.\right\}.$

If the previous proposition separates the languages of
$\Storage_2$ automata from those of $\Storage_1$ automata, then it
follows from the proof that
$\VAL(\Storage_2)$ is not accepted by any $\Storage_1$ automato (cf. \reflong{Appendix\ref{app:Languages}}).

\begin{corollary} \label{cor:HOCA-Language-Separation}
  If
  $
      \DTIME(\bigcup_{d\in\N}\exp_k(n^d)) \subsetneq
      \DSPACE(\bigcup_{d\in\N}\exp_k(n^d)) \subsetneq$\\
      $\DTIME(\bigcup_{d\in\N}\exp_{k+1}(n^d)),$
      then
      $L(\HOPS[(k-1)]) \subsetneq L(\HOCSO[k])
      \subsetneq$ \\
      $L(\HOCSM[k])
      \subsetneq L(\HOPS[k]).
$
\end{corollary}

The crucial underlying construction detail of the proof of
Proposition \ref{prop:EngelfrietLanguageReachConnection} is quite
hidden within the details of Engelfriet's  technical and long
paper. Its usefulness in other
contexts --- e.g., for higher-order pushdowns or counters ---  has been overseen
so far. Here we give another application to collapsible pushdown
automata:
reachability for collapsible pushdown automata of level~$k$
 is $\DSPACE(\exp_{k-1}(n))$-complete
(cf.~\cite{BroadbentCHS12}).
Thus, Proposition \ref{prop:EngelfrietLanguageReachConnection}
trivially shows that the language of valid level~$(k+1)$ collapsible
pushdown storage sequences separates the collapsible pushdown languages of
level~$k+1$ from those of level~$k$. This
answers a question asked by several experts in this field
(cf.~\cite{Parys12,Kobayashi13}).
In fact,
\cite{Parys12} uses a long and technical construction to prove the
weaker result that
there are more level $2k$ collapsible pushdown
languages than level $k$ collapsible pushdown languages.
From Proposition \ref{prop:EngelfrietLanguageReachConnection} one also easily
derives the level-by-level strictness of the collapsible pushdown tree
hierarchy and the collapsible pushdown graph hierarchy
(cf.~\cite{KartzowP12,Kobayashi13}).

\section{Future Work}
\label{sec:Conclusion}
Our result on regular reachability  gives hope
that also complexity results on model checking for logics like the
$\mu$-calculus extend from pushdown automata to \HOCS[2].
\HOCSO[2] probably is a generalisation of pushdown
automata that retains the good complexity results for basic
algorithmic questions.
It is also  interesting whether the result on
regular reachability extends to the different notions of
regularity for  \HOCS[k] mentioned in the
introduction.
\HOCS also can be seen as a new formalism in the context of register
machines
as currently used in the verification of concurrent systems.
\HOCS allow to store pushdown-like structures of register values and
positive results on model checking \HOCS could be transferred to
verification questions in this concurrent setting.

\clearpage
\appendix
\noindent\textsc{Note that there is an additional bibliographic
reference section at the end of the appendix; we use capital letters to
refer to these works, e.g., \cite{woehrle05}.}

\section{Omitted Proofs}
\label{sec:OmmitedProofs}

If $r$ is a run with domain $\{0, 1, \dots, m\}$ and $0\leq i \leq j \leq
m$ we write $r{\restriction}_{[i,j]}$ for the subrun from position $i$
to $j$, i.e., for the run $r'$ with domain $\{0, 1, \dots, j-i\}$ such
that  $r'(k) = r(i+k)$ for all $0\leq k \leq j-i$. 

\begin{proof}[of Lemma \ref{lem:BoundCounterForKReturns}]
  We first prove the claim for returns.
  Set $\hat m=\lvert \Sigma\rvert \lvert Q \rvert^2$.
  By induction on $k$, there is a $m_k\leq \hat m$ such
  that for all $k'\leq k$, all $\sigma\in\Sigma$ and all $m\geq m_k$
  $\ret{k'}((\sigma,m)) = \ret{k'}((\sigma,m_k))$ and
  $\left\lvert \bigsqcup_{\sigma\in\Sigma} \ret{k}((\sigma,m_k))
  \right\lvert \geq m_k$.\footnote{We use $\bigsqcup$ as the symbol
    for the disjoint union.}

  For the base case $k=-1$, let $m_{-1}=0$ and
  $\ret{-1}((\sigma,m)=~\emptyset$ for all $\sigma\in\Sigma$ and $m\in\N$.

  For the induction step note that for every $i\in\N$
  \begin{enumerate}
  \item
    $\bigsqcup_{\sigma\in\Sigma} \ret{k+1}((\sigma,i))$ contains at most
    $\lvert \Sigma \rvert \cdot \lvert Q\rvert^2$ many elements, and
  \item $\bigsqcup_{\sigma\in\Sigma} \ret{k}(\sigma,m_k) \subseteq
    \bigsqcup_{\sigma\in\Sigma} \ret{k+1}(\sigma,m_k)$.
  \end{enumerate}
  Since these sets are monotone in $i$,
  there is a  minimal number
  $m_k\leq m_{k+1}\leq \hat m$
  such that
  $\ret{k+1}((\sigma,m_{k+1}))=\ret{k+1}((\sigma,m_{k+1}+1))$ and
  $\bigsqcup_{\sigma\in\Sigma} \ret{k+1}((\sigma,m_{k+1}))$ contains at
  least $m_{k+1}$ elements.

  In order to complete our proof, we have to show that for all
  $n> m_{k+1}+1$,
  $\ret{k+1}((\sigma,n))\setminus
  \ret{k+1}((\sigma,m_{k+1}))=\emptyset$. Heading for a
  contradiction, assume that there
  is a minimal $n$ and a return $r$ witnessing that
  $(q,q')\in
  \ret{k+1}((\sigma,n))\setminus \ret{k+1}((\sigma,m_{k+1}))$.
  \begin{enumerate}
  \item If $r$ never visits a configuration of the form
    $r(j)=(q_j, s_j(\sigma_j, 0))$, we obtain 
    by monotonicity of $\PStorage(\CStorage)$  a run $r'$
    witnessing
    $(q,q')\in   \ret{k+1}((\sigma,n-1))\setminus
    \ret{k+1}((\sigma,m_{k+1}))$ contradicting minimality of $n$.
  \item Otherwise, there is a maximal $j$ such that
    $r(j)=(q_j, s_j(\sigma_j, 0))$. Since any operation alters the
    value of the topmost counter by at most $1$, we find a maximal
    $j_1\leq j$ such that $r(j_1)=(q_{j_1}, s_{j_1}(\sigma_{j_1},
    m_{k+1}+1))$. Since $r$ is a return and $j_1$ is not the last
    position in $r$, there is a $j_2\geq j$ such that
    $r(j_2)=(q_{j_2}, s_{j_1})$, i.e., the restriction of $r$ to
    $[j_1,j_2]$ is a return witnessing
    $(q_{j_1},q_{j_2})\in\ret{k'}((\sigma_{j_1},m_{k+1}+1))
    = \ret{k'}((\sigma_{j_1},m_{k+1}))$ for some $k'\leq k+1$.
    Due to 
    monotonicity of $\PStorage(\CStorage)$, we
    can lift a return from
    $(q_{j_1}, s_{j_1}(\sigma_{j_1},
    m_{k+1}))$ to state $q_{j_2}$ to a return $r'$ from
    $(q_{j_1}, s_{j_1}(\sigma_{j_1}, m_{k+1}+1))$ to state $q_{j_2}$
    such that the topmost counter of all configurations are at least
    $1$. Now replace in $r$ the subrun $r{\restriction}_{[j_1,j_2]}$
    by $r'$ and repeat this case distinction on the resulting run
    again.
  \end{enumerate}
  In the second case, we always choose a maximal $j$  such
  that the topmost counter is $0$. We then replace all occurring
  configurations  by others that do not assume the counter value $0$
  on the topmost  counter. Thus, if we iterate this process, the number
  $j$ in each step strictly decreases. Since the run is finite,
  after some iterations, we must reach the contradiction to the first
  case.

  Thus, we conclude that
  $\ret{k}((\sigma,m))=\ret{k}((\sigma,\hat m))$ for all $k\in\N$ and
  all $m\geq \hat m$. This immediately implies the analogous result for
  $k=\infty$.

  The claim for loops is proved completely analogous: there is a value
  $m_\infty$ between $\lvert \Sigma \rvert \lvert Q\rvert^2$ and
  $2\cdot \lvert \Sigma \rvert \lvert Q\rvert^2$ such that
  \begin{equation}
    \label{eq:loopsStabilise}
    \loops{\infty}((\sigma,m_\infty)) =
    \loops{\infty}((\sigma,m_\infty+1)).
  \end{equation}
  By a similar case distinction as in the return case, also
  from this point on the loops stabilise. The only difference now is
  that a counter value $0$ can occur  within a return starting
  with a topmost counter value $m_\infty+1$ or within a loop starting
  with topmost counter value $m_\infty+1$. Nevertheless either the
  first part of the lemma or \eqref{eq:loopsStabilise} allow to
  replace this subrun by one not visiting configurations with topmost
  counter value $0$. \qed
\end{proof}

\begin{proof}[of Lemma \ref{lem:ZerolessQtoFourReturnsareEnough}]
  Since for all $\sigma\in\Sigma$ and $i\in\N$ the sequence
  $\ret{k}((\sigma,i))$ is monotone in $k$,
  there is a number
  $k_0\leq \lvert \Sigma \rvert^2 \cdot \lvert Q\rvert^4$ such that
  $\ret{k_0}((\sigma,i))=\ret{k_0+1}((\sigma,i))$
  for all $1\leq i \leq \lvert \Sigma \rvert \cdot \lvert Q\rvert^2$
  and all $\sigma\in\Sigma$.
  Due to Lemma \ref{lem:BoundCounterForKReturns}, we conclude that for
  all $i\in\N$ and all $\sigma\in\Sigma$,
  $\ret{k_0}((\sigma,i))=\ret{k_0+1}((\sigma,i))$.

  Similar to the previous proof we now show that $\ret{k_0}=\ret{k}$
  for all $k\geq k_0$. For $k\leq k_0+1$, this is already guaranteed
  by choice of $k_0$.
  Assume that there are $\sigma\in\Sigma$, $k>k_0+1$ and $i\in\N$ such
  that $(q,q')\in \ret{k}((\sigma,i))\setminus\ret{k_0}((\sigma,i))$
  and that $r$ is a return witnessing this fact. We assume that $k$
  is minimal whence
  \begin{equation}
    \label{eq:retk=retk0}
    \ret{k_0}=\ret{k-1}.
  \end{equation}
  Thus, $r$ is a run that increases the
  height by $k$. It decomposes as $r=m_0\circ p_0 \circ r_0
  \circ m_1 \circ p_1 \circ r_1 \circ\dots\circ m_i \circ p_i \circ
  r_i \circ m_{i+1} \circ s$ where each $m_i$ is a subrun
  only using $\stay{f}$-operations (whence
  all configurations have the same height as the initial one),
  $p_i$ is a subrun performing only one $\push{\sigma_i,f_i}$, $r_i$ is a
  return, and $s$ is a subrun performing only one $\pop$-operation.
  Now, some of the $r_i$ increase the height by $k-1$.
  By \eqref{eq:retk=retk0}, we can replace each such $r_i$ by some
  return $r'_i$ that increases the height by at most $k_0$.
  This shows $(q,q')\in \ret{k-1}((\sigma,
  i))=\ret{k_0}((\sigma,i))$ contradicting our assumption.

  Thus, $\ret{k}=\ret{k_0}$ for all $k\geq k_0$ whence also
  $\ret{\infty}=\bigcup_{k\in\N}\ret{k} = \bigcup_{k\leq k_0}
  \ret{k} = \ret{\lvert\Sigma\rvert^2 \cdot \lvert Q \rvert^4}$.

  The proof for $\loops{\infty}=\loops{\lvert\Sigma\rvert^2 \cdot
    \lvert Q \rvert^4+1}$ follows because whenever a loop increases
  the height of the stack, it continues with some return. By the
  result for returns, this subreturn can be replaced by one that only
  increases the height by 
  $\lvert\Sigma\rvert^2 \cdot \lvert Q \rvert^4$.\qed
\end{proof}

\section{Reachability for \HOCSO[k] and \HOCSM[k]}
\label{app:ReachGeneral}
\subsection{Auxiliary Storage Automata}
Following Engelfriet's approach \cite{Engelfriet91}, we use
\emph{Auxiliary \SPACE{b(n)} \Storage automata} (where
$b:~\N\rightarrow\N$ is some function)  for the analysis of \Storage
automata. 
The former
are a general model
for computing. 
An instance is given by 
\begin{enumerate}
\item a finite control structure with control
  states $Q$,
\item an initial state $q_0\in Q$,
\item  transition rules $\Delta$,
\item a two-way read-only input tape,
\item a worktape (like for Turing machines) of size $b(n)$ where $n$
  is the size of the input, and
\item a storage \Storage.
\end{enumerate}
The storage \Storage
can be any known storage type used for automata, e.g., stacks, pushdowns, or
counters.  As usual, the above introduced automata can be deterministic,
nondeterministic, or alternating. 
We refer to \cite{Engelfriet91} for a detailed formal
introduction, the connection to Turing-machine based notions of time and space
complexity, as well as for references to the classical literature on
these machine models.

Note, that a \emph{1-way} auxiliary \SPACE{b(n)} \Storage automaton is
defined analogously whereas the input tape is only read one-way.
Most classical automata models can be directly rendered into this
framework, e.g., 
nondeterministic 1-way \Storage automata where \Storage is the trivial
storage correspond to nondeterministic finite automata;  2-way
auxiliary \SPACE{b(n)}  $\Storage$ automata where $\Storage$ is the
trivial storage are the classical
$b(n)$-space bounded Turing machines; 1-way
\Storage automata
where \Storage is a pushdown  correspond to pushdown
automata. 

The configuration of an auxiliary \SPACE{b(n)} \Storage automaton is the
tuple containing the current finite state $q\in Q$,
the contents of the auxiliary work tape, i.e., a word $w$ of size
bounded in \SPACE{b(n)}, as well as the configuration $\Sconfig\in
\SConfig$ of \Storage.
As usual, we define a run of an automaton as a sequence of
configurations that is conform with the underlying transition rules
and the semantics of the storage type. The applicable transition rules
depend on the outcome of the storage tests applied to the current
storage configuration, the current control state and the next input
symbol.

\subsection{Technical Results}

Before we analyse control-state reachability on \HOCSM[k] and \HOCSO[k], we recall
and extend some results of Engelfriet.
The following results are Theorems 2.2 and 2.4 of
\cite{Engelfriet91}.

\begin{lemma}\label{lem:Engelfriet2.2}
  Let $b$ be some function satisfying  $b(n)\geq\log(n)$ for all
  $n\in\N$ and let $\Storage$ be a storage type.
  In polynomial time, we can translate
  an
  alternating auxiliary $\SPACE{ \bigcup_{d\in\N} \exp(d b(n))}$
  \Storage automaton  into an
  alternating auxiliary $\SPACE{b(n)}$ {\PStorage}(\Storage)
  automaton such that both automata accept the same language and vice versa.
\end{lemma}

\begin{lemma}\label{lem:Engelfriet2.4}
  For $b(n)\geq\log(n)$ and every storage type $\Storage$,
  there are polynomial time algorithms that translate
  a nondeterministic auxiliary $\SPACE{b(n)}$ ${\PStorage}(\Storage)$
  automaton into an
  alternating auxiliary $\SPACE{b(n)}$
  \Storage automaton accepting the same language and vice versa.
\end{lemma}

A detailed look on Engelfriet's proof of Lemma 7.11 in~\cite{Engelfriet91} reveals
that its analogue for alternating automata holds if we replace
the role of nonemptyness by the role of control state reachability.
Moreover, the
logspace reduction of membership for auxiliary $\SPACE{log(n)}$
\Storage automata
to control state reachability for 1-way $\Storage$ automata
extends to an
$b(n)$ space-bounded reduction of
membership for auxiliary $\SPACE{b(n)}$
\Storage automata
to control state reachability for $\Storage$ automata (which now may
have size $\exp(b(n))$ when starting with an input of size $n$).
Before we prove these claims, let us first define alternating
reachability in our setting.

\begin{definition}
  Let $\Aa$ be an alternating auxiliary $\SPACE{b(n)}$ $\Storage$
  automaton. For a
  set $C$ of configurations of $\Aa$ we define
  $\mathsf{pre}^*(C)$ as the set of configurations $c$ such that there
  is a computation tree $T$ of $\Aa$
  \begin{enumerate}
  \item the root of $T$ is labelled by $c$,
  \item all leaves of $T$ are labelled by configurations $c'\in C$,
  \item for each inner node $t$ of $T$ labelled by an existential configuration
    $c$  there is exactly one successor $t'$ in $T$ and $t'$ is labelled by a
    successor $c'$ of $c$ (w.r.t $\Aa$), and
  \item for each inner node $t$ of $T$ labelled by a universal configuration
    $c$ there is, for each  successor $c'$ of $c$ (w.r.t $\Aa$) a
    successor $t'$ of $t$ labelled by $c'$.
  \end{enumerate}
  For a state $q$ of $\Aa$, we say that $q$ is  alternatingly
  (control state) reachable in $\Aa$ if the
  initial configuration of $\Aa$ belongs to
  $\mathsf{pre}^*(\{(q,\Sconfig)\mid \Sconfig$ is an
  $\Storage$-configuration$\})$.
\end{definition}
\begin{remark}
  Note that we disallow that the computation tree $T$ may contain a
  leaf labelled by some universal state $c$ not in $C$ such that no
  transition of $\Aa$ is applicable to $c$. 
  This restrictive definition is necessary for the results provided in
  the following. 
\end{remark}

\begin{definition}
  The alternating control state reachability problem for some class
  $\mathcal{C}$ of automata is the
  following.\\
  Input: $\Aa\in\mathcal{C}$, $q$ a state of $\Aa$\\
  Output: Is $q$ alternatingly reachable in $\Aa$.
\end{definition}

The following lemmas extend Engelfriet's result on the connection
between emptiness (or equivalently, control state reachability) of
nondeterministic $\Storage$ automata and membership for
nondeterministic auxiliary
$\SPACE{\log(n)}$ \Storage automata to the setting of alternating
automata.
\begin{lemma}\label{lem:AltAutomataReducibleToAlternatingLogSpace}
  Let $\Storage$ be a  storage type.
  Alternating (or nondeterministic, respectively)
  control state reachability of (1-way) \Storage automata  reduces
  to membership of alternating (nondeterministic,
  respectively) auxiliary $\SPACE{\log(n)}$ \Storage automata via
  logspace reductions.
\end{lemma}
\begin{proof}
  Let $A(\Storage)$ denote the set of 1-way alternating  \Storage
  automata  and fix an effective encoding $\overbar{A}(\Storage)$ of this
  set as binary strings. We
  write $\overbar{\Mm} \in \bar A(\Storage)$ for the encoding of the
  automaton $\Mm \in A(\Storage)$.  Analogously, we write
  $\overbar{q}$ for
  the encoding of some state $q$.
  We define an auxiliary $\SPACE{\log(n)}$ \Storage automaton
  which we call $\Aa$ such that
  $L(\Aa)=\{\overbar{\Mm}\#\overbar{q} \mid \Mm\in A(\Storage), q$ a state of $\Mm$ and
  $q$ alternating reachable by $\Mm\}$.
  Given an input string $s$, $\Aa$ first checks that
  $s=\overbar{\Mm}\#\overbar{q}$ for some $\Mm\in A(\Storage)$ and some state $q$ of $\Mm$. Now $\Aa$
  simulates $\Mm$
  storing two pointers on its tape, one called state pointer and one
  called transition pointer.

  As initialisation the state pointer is set to the position of
  the input where the initial state of $\Mm$ is encoded.

  We now iterate the following case distinction.
  If the state pointer points to (the encoding of) $q$, $\Aa$
  accepts.
  Otherwise, the state pointer points to some state $q'$. Scanning the
  input string we determine whether $q'$ is an existential state of
  $\Mm$. If this is the case, we do an existential simulation step,
  otherwise we do a universal simulation step.
  \begin{itemize}
  \item \emph{Existential simulation step.}
    The state pointer points to some state $q'$.
    Now we guess a transition  applicable to the current configuration
    of $\Mm$ (which is $(q',\Sconfig)$ for $\Sconfig$ the current
    storage configuration of $\Aa$.
    This is done by setting the transition pointer to some value $i$ such
    that at position $i$ in the input string the encoding of a
    transition $\delta=(p, t, f, r$) starts.
    $\Aa$ now checks that $p=q'$. Then it checks that the test formula
    $t$ is satisfied by the current storage configuration. If this is
    not the fact, $\Aa$ rejects. Otherwise it applies $f$ to the
    storage and changes the state pointer such that it points to the
    encoding of $r$.
  \item \emph{Universal simulation step.}
    The state pointer points to some universal state $q'\neq q$ and
    the current $\Storage$ configuration (of $\Aa$ and of $\Mm$) is
    $\Sconfig$.
    Recall that $q$ is alternatingly reachable from $(q',\Sconfig)$ if
    there is a computation tree where the root is labelled by
    $(q',\Sconfig)$ and is not a leaf (because $q\neq q'$) and $q$ is
    alternatingly reachable
    from every successor of $(q',\Sconfig)$ in the computation tree.

    In order to guarantee that $(q',\Sconfig)$ has a successor
    configuration with respect to $\Mm$, we universally spawn a
    subprocess that performs an existential simulation step.
    If this branch accepts, we still have to show
    that for any applicable transition, $q$ is alternatingly reachable
    from the resulting configuration.
    For this purpose the transition pointer iterates over all
    positions in the encoding $\overbar{\Mm}$ of $\Mm$.
    As soon as this iteration has been finished, this main process
    accepts. During the iteration it may spawn subprocesses as
    follows.

    If the current position of the pointer points to a transition
    $\delta=(p, t, f, r)$ we check whether $q'=p$. In this case we
    universally spawn a subprocess. It
    checks whether the test formula $t$ is satisfied by the current
    storage configuration. If not, the process accepts. Otherwise,
    $\Aa$ universally branches to an accepting branch and another branch
    by first applying $f$ to the current storage configuration and
    then setting the
    state pointer to the position of the encoding of $r$ and starting
    the next simulation step.
    It is straightforward to see that one of the following holds.
    \begin{enumerate}
    \item $f$ is not applicable to the current storage configuration,
      thus $\delta$ does not provide a successor of the current
      configuration of $\Mm$. In this case, the universal branching
      only spawns one accepting branch whence the subprocess dealing
      with $\delta$ accepts.
    \item $f$ is applicable to the current storage configuration. Then
      the subprocess applying $f$ accepts if and only if $q$ is
      alternatingly reachable from the
      $\delta$-successor of the current configuration.
    \end{enumerate}
  \end{itemize}
  It is straightforward to prove that $\Aa$ accepts
  $\overbar{\Mm}\#\overbar{q}$
  if and only if $q$ is alternatingly reachable by $\Mm$.
  Moreover, $\Aa$ only needs universal states for the universal
  simulation step. Thus, restricting the input to nondeterministic
  $\Storage$ automata, the simulating machine $\Aa$ will also be
  nondeterministic instead of alternating. \qed
\end{proof}

Engelfriet also provided a logspace reduction in the other direction
in the nondeterministic case. We extend this result again to the
alternating case and to auxiliary $\SPACE{b(n)}$ \Storage automata for
arbitrary space bound $b$.

\begin{lemma}\label{lem:spaceboundedTMtoReachabilityofAutomata}
  Let $b(n)\geq \log(n)$ and $\Mm$ be an
  alternating (or nondeterministic) auxiliary
  $\SPACE{b(n)}$ \Storage automaton.
  The membership problem for $\Mm$ is
  reducible to
  alternating (nondeterministic, respectively) control state
  reachability  for  \Storage automata via a
  $\DSPACE(b(n))$-computation.
\end{lemma}
\begin{proof}
  Let $\Gamma$ be the tape alphabet with blank symbol $\Box$,
  $\Sigma$ the input alphabet,
  $Q$ the state set and
  $q_0\in Q$ the initial state
  of $\Mm$.
  Without loss of generality $\Mm$ has only $1$ accepting state
  and it enters this state if and only if the tape is completely blank
  and the heads of the input and the reading tape are on the first
  cell.

  On input a word $w$,
  we construct an automaton
  $\Aa$ with state set
  $Q\times \Gamma^{b(\lvert w \rvert)} \times
  \{1,2,\dots b(\lvert w \rvert)\} \times \{1,2,\dots, \lvert w
  \rvert\}$.
  Note that each configuration fits into space $O(b(\lvert w
  \rvert))$.
  The initial state is
  $(q_0, w\Box^{b(\lvert w \rvert)-\lvert w \rvert},1,1)$.
  Some state $c=(q,\gamma_1\dots \gamma_{b(\lvert w \rvert)}, i,
  j)$ represents the configuration of $\Mm$ where the work
  tape contains the letters $\gamma_1\dots \gamma_{b(\lvert w
    \rvert)}$, $\Mm$ is in state $q$ the head of the work tape
  is at position $i$ and the head of the input tape is at position
  $j$. This state is an existential one if and only if  $q$ is an
  existential state of $\Mm$. $\Aa$ has a transition
  $(c,t,f,c')$
  to state
  $c'=(q', \gamma_1'\dots, \gamma_{b(\lvert w \rvert)}', i', j')$
  if and
  only if $\Mm$ has a transition with test-formula $t$ and storage
  operation $f$  whose application would translate configuration $c$
  to configuration  $c'$ (for all storage configurations  where $t$ is
  satisfied and $f$ is applicable).
  The final state of $\Aa$ is
  $c_f=(q_f, \Box^{b(\lvert w \rvert)}, 1, 1)$.

  It is straightforward to prove that $c_f$ is alternatingly reachable
  by $\Aa$  if
  $\Mm$ accepts $w$.
  Note that $\Aa$ contains universal states if and only if $\Mm$
  contains universal states.\qed
\end{proof}

Analogously to Engelfriet's proof that a pushdown can replace
alternation, we now investigate tradeoffs concerning
the storage type $\CStorageZero$.
This proof is inspired by the  $\PSPACE$-hardness proof for emptiness
of alternating finite automata recently published by Jancar and Sawa \cite{JancarS07}.

\begin{lemma}\label{lem:DSPACE-ALOGSPACE+C}
  Let $b(n)=\exp_{k}(n)$ for some $k\geq 0$.
  Let $\Mm$
  be a deterministic auxiliary $\SPACE{b(n)}$ automaton, i.e., a
  deterministic  $\SPACE{b(n)}$ Turing-machine.
  We can compute in logspace 
  an
  alternating auxiliary $\SPACE{\log(b(n))}$ \CStorageZero automaton
  $\Aa$ such that $\Mm$ accepts $w$ iff $\Aa$
  accepts $w$ for all $w\in\Sigma^*$.
\end{lemma}
\begin{proof}
  Assume that
  $\Mm$ has state set $Q$,
  initial state $q_0\in Q$,
  final state   $q_f\in Q$  and
  tape alphabet  $\Gamma$.
  The main states of $\Aa$ come from the set
  $Z=\Gamma\cup(\Gamma\times Q)$.
  Moreover the state set $Q'$ of $\Aa$ contains $p(\lvert Z\rvert)$
  many auxiliary states for some polynomial $p$. For simplicity of the
  presentation we omit the formal specification of these states.
  Our goal is to construct an automaton $\Aa$
  whose configurations are of the form   $(z,t,i)\in Z\times
  \{0,1\}^{\log(b(n))}\times \N$ where
  $z$ is the current state of
  $\Aa$,
  $t$ is the content of its tape (which we identify with a
  binary encoded natural number between $0$ and $b(n)$) and
  $i$ is the current counter value.

  Our goal is to define $\Aa$ in such a way that
  $\Aa$ accepts from configuration
  $(z,t,i)$
  on
  input $w$ if at time step $i$ of the computation of $\Mm$ at the
  $t$-th cell of $\Mm$'s tape, the content is $z$ (where we say
  that the $t$-th cell content is $(q,\gamma)\in Q\times \Gamma$ if
  the cell contains $\gamma$ and $\Mm$ is reading this cell in state $q$).
  Let $\mathsf{preds}(z)$ be the set of triples $(z_1,z_2,z_3)$
  such that the one-step computation of $\Mm$ on the tape described by
  $z_1z_2z_3$ leads to the replacement of $z_2$ by $z$.
  If $\Aa$ is in some configuration $(z,t,i)$ with $0 < t < b(n)$ and
  $i>0$, it
  nondeterministically chooses the hopefully correct triple
  $(z_1,z_2,z_3)\in\mathsf{preds}(z)$ and universally branches to
  configurations $(z_1,t-1,i-1), (z_2,t,i-1), (z_3,t+1,i-1)$.
  Note that a finite amount of auxiliary states suffices to calculate
  the tape content $t+1$ and $t-1$ from $t$.
  We now specify the acceptance condition.
  Configurations $(\Box, 0, i)$ and
  $(\Box, b(n), i)$ are accepting (for all $i\in\N$)
  while all other configurations with tape $t=0$ or
  $t=b(n)$ are rejecting (again, only finitely many states are
  needed to check whether we are in one of these configurations).
  Assuming that the input is $w=a_1\dots a_n$,
  let configuration
  $((q_0,a_1), 1, 0)$,  configurations
  $(a_i, i, 0)$ for $2 \leq i \leq n$, and configurations
  $(\Box, j, 0)$ for $j>n$
  be all accepting.
  All other configurations with counter value $0$ are
  rejecting. Note that this acceptance condition relies on the input
  and can be checked with finitely many auxiliary states.
  In a configuration $(z,t,0)$ we parse the input word to the
  $t$-th letter and compare z to this letter (if $w$ ended before,
  then $z$ has to be the blank symbol $\Box$).

  An easy induction on $i$ shows that there is an accepting
  computation of $\Aa$
  starting in  $(z,t,i)$ with input $w$ if and only if
  in the computation of $\Mm$ on $w$ the $t$-th letter
  of the $i$-th configuration is $z$ (where as before
  $z=(q,\gamma) \in Q\times \Gamma$ means that
  the content of the  $t$-th cell is $\gamma$,
  $\Mm$'s head is positioned at the $t$-th cell and
  $\Mm$ is in state $q$).

  Now we add to $\Aa$ an initialisation phase that,
  on input $w$, guesses
  a letter $\gamma\in\Gamma$,
  a number $t\leq b(\lvert w \rvert)$ and
  some number $i\in\N$ switching to
  configuration
  $((q_f,\gamma), t, i)$. Now $\Aa$ accepts $w$ if and only if the
  computation of $\Mm$ on $w$ is accepting.\qed
\end{proof}

\subsection{Control-State Reachability on $\HOCSM[k]$}
\label{app:applications}

We prove the part of Theorem \ref{thm:HOCSM-Reachability} on 
reachability. The claim for alternating reachability follows directly
from this result as we will explain in Section \ref{sec:altCont}.
We determine the exact complexity of reachability on
\HOCSM[k].
For the base case we use a result mentioned by
G\"oller \cite{Goller08}.

\begin{lemma}\label{lem:PSPACE-c-of-ReachCPlus}
  Alternating control state reachability for alternating
  (1-way) $\CStorageZero$ automata is $\PSPACE$-complete.
\end{lemma}

\begin{proposition}
  Control state reachability for \HOCSM[k] is
  $\DSPACE(\bigcup_{d\in\N} \exp_{k-2}(n^d))$-complete  for all $k\geq 2$.
\end{proposition}
\begin{proof}
  For containment,
  let us first consider the case $k=2$.
  Given a $2$-\HOCSM $\Aa$  and a state $q$, control state
  reachability reduces by
  Lemma \ref{lem:AltAutomataReducibleToAlternatingLogSpace}
  to a membership
  problem for a (nondeterministic) auxiliary (2-way) $\SPACE{\log(n)}$
  ${\PStorage}(\CStorageZero)$
  automaton.
  Due to Lemma \ref{lem:Engelfriet2.4} this automaton
  can be translated into
  an alternating auxiliary $\SPACE{\log(n)}$ $\CStorageZero$
  automaton.
  Due to Lemma
  \ref{lem:spaceboundedTMtoReachabilityofAutomata}, membership for
  this machine is logspace reducible
  to alternating control state reachability on  alternating (1-way)
  $\CStorageZero$ automata which by
  Lemma \ref{lem:PSPACE-c-of-ReachCPlus} is solvable in
  $\PSPACE = \DSPACE(\bigcup_{d\in\N} \exp_{0}(n^d))$.

  Now we proceed by induction on $k$.
  Given a \HOCSM[k] $\Aa$ ($k\geq 3$) and a state $q$, control state
  reachability reduces by
  Lemma \ref{lem:AltAutomataReducibleToAlternatingLogSpace}
  to a membership
  problem for a (nondeterministic)
  auxiliary (2-way) $\SPACE{\log(n)}$ $\PStorage^{k-1}(\CStorageZero)$
  automaton.
  Due to Lemma \ref{lem:Engelfriet2.4} this machine can be translated into
  an alternating auxiliary $\SPACE{\log(n)}$
  $\PStorage^{k-2}(\CStorageZero)$ automaton.
  We apply Lemma \ref{lem:Engelfriet2.2} and obtain an equivalent
  alternating auxiliary $\SPACE{n^d}$
  $\PStorage^{k-3}(\CStorageZero)$ automaton (for some $d\in\N$).
  Again with Lemma \ref{lem:Engelfriet2.4} this is
  translated to a nondeterministic auxiliary
  $\SPACE{n^d}$
  $\PStorage^{k-2}(\CStorageZero)$ automaton.
  Using the polynomial-space reduction from
  Lemma \ref{lem:spaceboundedTMtoReachabilityofAutomata} we obtain a
  state $q'$ and a
  \HOCSM[(k-1)] $\Aa'$ of size $\exp(O(\lvert \Aa\rvert^d))$ such that $q$
  is reachable in $\Aa$ if and only if $q'$ is
  reachable in $\Aa'$. By induction hypothesis the latter is
  decidable
  in space $\exp_{k-3}(\lvert \Aa' \rvert^{d'})$ for some
  $d'\in\N$. Thus, in terms of $\lvert \Aa \rvert$
  the space is bounded by
  $\exp_{k-3}( (\exp( O(\lvert \Aa \rvert)^d)^{d'})) =
  \exp_{k-2}(O(\lvert \Aa \rvert)^d)$.
  This completes the containment proof.

  We now prove hardness.
  Recall that Lemma  \ref{lem:DSPACE-ALOGSPACE+C} provided a reduction
  of any membership problem in
  $\DSPACE(\exp_{k-2}(n^d))$ ($d\in\N$) to a membership problem for
  an alternating auxiliary $\SPACE{\exp_{k-3}(n^d)}$ \CStorageZero automaton.
  Due to Lemma \ref{lem:Engelfriet2.2} this can be reduced
  to a membership problem for an
  alternating auxiliary $\SPACE{\log(n^d)}$
  $\PStorage^{k-2}(\CStorageZero)$ automaton. Furthermore, by Lemma
  \ref{lem:Engelfriet2.4} this reduces to a membership problem for a
  nondeterministic auxiliary $\SPACE{\log(n^d)}$
  $\PStorage^{k-1}(\CStorageZero)$ automaton. Finally,
  due to Lemma
  \ref{lem:spaceboundedTMtoReachabilityofAutomata},
  there is a polynomial time reduction of this problem to
  a control state
  reachability problem for a (1-way) $\PStorage^{k-1}(\CStorageZero)$
  automaton of size $O(n^{d})$,
  i.e., reachability for \HOCSM[k].\qed
\end{proof}

\subsection{Control-State Reachability on \HOCSO[k]}
\label{sec:subsechocsoReach}

Based on our result that control state reachability for $\HOCSO[2]$
is in
$\PTIME$ (Proposition\,\ref{prop:ReachHOCS-0}), Engelfriet's machinery allows to
determine the complexity of reachability in $n$-\HOCSO
inductively. This proves the first half of Theorem
\ref{thm:HOCSO-Reachability}. The claim on alternating reachability is
proved in the following section. 

\begin{proposition}\label{prop:HOCSO-Reachability}
  For $k\geq 2$,
  the control state reachability problem for \HOCSO[k] is in
  $\DTIME(\bigcup_{d\in\N}\exp_{k-2}(n^d))$.
\end{proposition}
\begin{proof}
  We use Engelfriet's machinery
  and induction:
  the case reachability for $\HOCSO[2]\in \PTIME$ has already been shown
  in Proposition\,\ref{prop:ReachHOCS-0}.
  Given a \HOCSO[k] $\Aa$ of level $k\geq 3$ and a state $q$, control state
  reachability reduces by
  Lemma \ref{lem:AltAutomataReducibleToAlternatingLogSpace}
  to a membership
  problem for some nondeterministic auxiliary $\SPACE{\log(n)}$
  $\PStorage^{k-1}(\CStorage)$ automaton.
  Due to Lemma \ref{lem:Engelfriet2.4} this automaton can be translated into
  an alternating auxiliary $\SPACE{\log(n)}$
  $\PStorage^{k-2}(\CStorage)$ automaton.
  We apply Lemma \ref{lem:Engelfriet2.2} and obtain an
  alternating auxiliary $\SPACE{d n}$  $\PStorage^{k-3}(\CStorage)$
  automaton for some $d\in\N$.
  Again with Lemma \ref{lem:Engelfriet2.4} this is
  translated to a nondeterministic auxiliary $\SPACE{d n}$
  $\PStorage^{n-2}(\CStorage)$ automaton.
  Finally, we use the polynomial-space reduction from
  Lemma \ref{lem:spaceboundedTMtoReachabilityofAutomata} and obtain a
  state $q'$ and a
  \HOCSO[(k-1)] $\Aa'$ of size exponential in that of
  $\Aa$ such that $q$ is reachable in $\Aa$ if $q'$ is
  reachable in $\Aa'$.
  By induction hypothesis, we can decide this
  in time
  $\exp_{n-3}(p'(\lvert \Aa' \rvert))
  =
  \exp_{n-2}(p(\lvert \Aa \rvert))$  for some polynomials $p$ and $p'$.\qed
\end{proof}

\subsection{Alternating Control State Reachability}
\label{sec:altCont}

We derive our results on alternating reachability by use of a much
more general relation between the pushdown operator and alternation.

\begin{proposition}
  Given any storage type $\Storage$, alternating control state
  reachability for $\Storage$ automata is logspace reducible to
  control state reachability of $\PStorage(\Storage)$ automata and
  vice versa.
\end{proposition}
\begin{proof}
  Let $\Aa$ be an alternating $\Storage$ automaton and $q$ some
  state. By Lemma \ref{lem:AltAutomataReducibleToAlternatingLogSpace}
  the alternating control state reachability problem for $(\Aa, q)$
  reduces to a membership problem for an alternating auxiliary
  $\SPACE{\log(n)}$ $\Storage$ automaton. This reduces by 
  Lemma \ref{lem:Engelfriet2.4} to a membership problem for a
  nondeterministic auxiliary $\SPACE{\log(n)}$ $\PStorage(\Storage)$
  automaton.
  Finally, using Lemma
  \ref{lem:spaceboundedTMtoReachabilityofAutomata} this problem
  reduces to a control state reachability problem for a
  nondeterministic $\PStorage(\Storage)$ automaton. 
  
  Using the 'nondeterminism' variant of Lemma
  \ref{lem:AltAutomataReducibleToAlternatingLogSpace}, the other
  direction of Lemma \ref{lem:Engelfriet2.4} and the 'alternation'
  variant of Lemma \ref{lem:spaceboundedTMtoReachabilityofAutomata},
  the control state reachability problem for $\PStorage(\Storage)$
  automata similarly reduces to the alternating control state
  reachability problem for alternating $\Storage$ automata. 
\end{proof}

\section{Equivalence of Storages}
\label{app:Storage_Equivs}
In \cite{Slaats12} the notion of a level $k$ counter automaton with
$0$-test was defined differently from our notion of $\HOCSM[k]$ as
follows. Basically Slaats uses the storage type
$\PStorage_{\{\bot\}}^{k-1}(\CStorageZero)$ instead of
$\PStorage_{\{\bot,0,1\}}^{k-1}(\CStorageZero)$.
In the following we show that both variants lead to equivalent automata.
Let us first  recall the notion of equivalence of storage types
(cf.~\cite{Engelfriet91}).

\begin{definition}
  Let $\Storage$ and $\Storage'$ be storages.
  $\Storage$ can \emph{simulate}
  $\Storage'$, denoted as $\Storage' \preceq \Storage$,
  if
  for every
  one-way deterministic $\Storage'$ transducer there is a one-way
  deterministic $\Storage$ transducer defining the same transductions.

  $\Storage$ and
  $\Storage'$ are \emph{equivalent},
  denoted as $\Storage \equiv \Storage'$, if
  $\Storage \preceq \Storage'$ and $\Storage'\preceq \Storage$.
\end{definition}
\begin{remark}
  As pointed out by Engelfriet,
  for storage types $\Storage, \Storage'$ such that
  $\Storage\preceq\Storage'$, for
  $t\in\{\text{nondeterministic, alternating, deterministic}\}$
  $t$ $\Storage$ automata can be simulated by
  $t$ $\Storage'$automata.
\end{remark}

Recall that we defined the storage type
$\CStorage = \PStorage_{\{\bot\}}$. In the following, we also use
$\CStorage$ as the operator $\PStorage_{\{\bot\}}$ acting on other storage
types. We call $\CStorage(\Storage)$ the storage type \emph{counter of
  $\Storage$}.

\begin{proposition}\label{prop:equivStorages1}
  It holds that $\CStorage^{k-1}(\CStorageZero) \equiv
  \PStorage^{k-1}(\CStorageZero)$.
\end{proposition}
\begin{proof}
  The direction from left to right is clear because $\PStorage$ is an
  extension of $\CStorage$.
  We  show how $\CStorage^{k-1}(\CStorageZero)$ can simulate
  $\PStorage^{k-1}(\CStorageZero)$.

  We first show that $\Storage:=\PStorage(\CStorage^{k-2}(\CStorageZero))$
  can be simulated by $\Storage':=\CStorage^{k-1}(\CStorageZero)$. The
  idea is to
  encode the pushdown symbol of level $k$, by the level 1 counter
  value modulo $3$ (recall that $\PStorage$ uses the  pushdown
  alphabet $\{\bot, 0, 1\}$).
  For this purpose we first replace in the
  $\PStorage(\CStorage^{k-2}(\CStorageZero))$ every $\push{\bot}$ of
  level $1$(i.e. a push applied to $\CStorageZero$) by $3$
  $\push{\bot}$ operations and each level $1$ $\pop$-operation by $3$
  $\pop$-operations of level $1$.
  This results in an equivalent $\Storage$ automaton where the level
  $1$ counter value is always $0 \mod 3$. Next, without loss of
  generality we assume that the $\Storage$ automaton only uses
  instructions of the form $\pop, \push{\sigma,\id}$ and $\stay{f}$.
  \emph{For the rest of this simulation, we identify $\bot$ with the
    number $2$}. We want to represent a
  pushdown symbol $\sigma\in\{0,1,\bot\}$ by $\sigma \mod 3$ on the
  level $1$ counter.
  We initialise $\Storage'$ by applying $2$ $\push{\bot}$ on level
  $1$ (this results in the counter value $2$, which is  $2\mod 3$
  representing the initial symbol $\bot$. Now we simulate the
  operations on $\Storage$ by $\Storage'$-operations as follows (where
  we assume that the current $\Storage$-configuration $\Sconfig$ is
  simulated by $\Storage'$-configuration $\Sconfig'$.
  \begin{enumerate}
  \item The $\TOP{\gamma}$ test for $\gamma\in\{0,1,\bot\}$ can  be
    simulated as follows.
    apply $\push{\bot,\id}$, then
    determine the topmost symbol
    $\gamma'\in\{0,1,\bot\}$
    by level $1$
    $\pop$-operations (while the $0$-test fails) determining the value
    of the topmost level $1$
    counter modulo $3$. After finishing the test we restore the
    pushdown by a $\pop$ operation and just have to compare $\gamma$
    with $\gamma'$.
  \item The $\Temptystack$ test on level $1$ is simulated by first
    determining which $\TOP{\gamma}$ test applies for
    $\gamma\in\{0,1,\bot\}$ as in the simulation of
    $\TOP{\gamma}$. Then we perform
    $\gamma$ many
    $\pop$-operations of level $1$, then the $\Temptystack$ test of
    $\Storage'$  coincides with the $\Temptystack$ of $\Storage$. We
    restore the pushdown by $\gamma$ many $\push{\bot}$ operations of
    level $1$.
  \item A $\pop$ operation is simulated by $\pop$.
  \item A $\push{\gamma,\id}$ operation is simulated by the following
    program:
    first determine the topmost symbol $\gamma'\in\{0,1,\bot\}$
    of $\Sconfig'$. Then apply $\push{\bot,\id}$, then apply $\gamma'$
    many level $1$ pop-operations.  No we apply $\gamma$ many level
    $1$ $\push{\bot}$ operations.
  \item A $\stay{f}$ operations is simulated by $\stay{f}$ if $f$ is
    not an operation of level $1$. If it is of level $1$ we just
    duplicate it $3$ times.
  \end{enumerate}

  This completes the proof that
  $\PStorage(\CStorage^{k-2}(\CStorageZero))$ can be simulated by
  $\CStorage^{k-1}(\CStorageZero)$. The lemma now follows
  by induction on $k$:
  we have shown that
  $\PStorage^1(\CStorageZero) \equiv \CStorage^1(\CStorageZero)$.
  Assume that for some $k$ we have $\PStorage^{k-1}(\CStorageZero)
  \equiv \CStorage^{k-1}(\CStorageZero)$.
  By Theorem 1.3.1  of \cite{Engelfriet91},
  we obtain
  \begin{equation*}
    \PStorage^k(\CStorageZero) \equiv
    \PStorage( \PStorage^{k-1}(\CStorageZero)) \equiv
    \PStorage( \CStorage^{k-1}(\CStorageZero)) \equiv
    \CStorage^k(\CStorageZero).
  \end{equation*}\qed
\end{proof}

Readers interested in a more throughout comparison of different
possible definitions of higher-order one-couter automata are invited
to have a look at Appendix \ref{app:Storage_Equivs_Long}.

\section{Separation of Languages of Higher-Order Counter Automata
 With or Without $0$-test}

\label{app:Languages}

Under the assumption that
\begin{align*}
  \DTIME(\bigcup_{d\in\N}\exp_{k}(n^d)) \subsetneq
  \DSPACE(\bigcup_{d\in\N} \exp_{k}(n^d))\subsetneq
  \DTIME(\bigcup_{d\in\N} \exp_{k+1}(n^d))
\end{align*}
our results on
the reachability problem for \HOCS implies a strict separation of the
languages of higher-order counters and higher-order pushdowns.

We first recall some results of Engelfriet that allows to shift
results on $2$-way auxiliary automata down to $1$-way automata.
We recall his proofs in order to extract the constructive
content.

Let
N-aux-$\SPACE{\log(n)}-\Storage$-L
denote the languages accepted by
nondeterministic auxiliary $\SPACE{\log(n)} \Storage$ automata.
Let $T$ be the class of nondeterministic
logspace transducers.
Let
$T^{-1}(\mathcal{L}):=\set{\tau^{-1}(L)}{\tau\in T,L\in\mathcal{L}}$
be the class of languages obtained by application of transductions
from $T$ to languages from $\mathcal{L}$.

Recall that $\VAL(\Storage)$ is the language of valid storage
sequences for storage type $\Storage$.
It is accepted by a deterministic
$\Storage$ automaton with only $1$ state $q$ and no $\varepsilon$-transitions
that works as follows.
Transitions on input an $\Storage$-operation $f$ are of the form
$(q,f, \emptyset, f, q)$, i.e., $\Ss$ on input $f$ applies $f$
unconditionally and transitions on input a test
 $(t=r)$ are of the form $(q, (t=r), (t,r), \id_{\SConfig}, q)$< i.e.,
computation continues if test $t$ results in $r$ and the storage
remains unchanged.

\begin{lemma}[\cite{Engelfriet91}, Lemma 7.1]
  For every storage type $\Storage$,
  N-aux-$\SPACE{\log(n)}-\Storage$-L $=
  T^{-1}(1N-\Storage)=T^{-1}(1D-\Storage) = T^{-1}(\{\VAL(\Storage)\})$
  where
  $1N-\Storage$ ($1D-\Storage$, respectively) denotes the class of languages
  accepted by nondeterministic (deterministic, respectively)
  $\Storage$-automata.
\end{lemma}
\begin{proof}[sketch]
  Given a nondeterministic auxiliary $\SPACE{\log(n)}-\Storage$
  automaton $\Aa$ we can split it into two devices as follows.
  First we use a nondeterministic logspace
  transducer $\Tt$ that simulates $\Aa$ but instead of performing
  $\Storage$-tests or -operations  it writes these on the output tape
  (tests are written together with the expected test result). Then we
  use the deterministic $\Storage$ automaton $\Ss$  recognising
  $\VAL(\Storage)$ and check whether the output of $\Tt$ is a valid
  sequence of operations and
  tests of $\Storage$.

  For the other direction, given a transducer and an
  $\Storage$-automaton, the language of their composition is
  recognised by a
  nondeterministic auxiliary $\SPACE{\log(n)}-\Storage$
  automaton which is a simple product of the two automata. \qed
\end{proof}
A straightforward extension of Engelfriet's Corollary 7.2 from
\cite{Engelfriet91} is the following.

\begin{corollary}\label{cor:Engelfriet7.2}
  Let $\Storage$ and $\Storage'$ be storage types. If
  N-aux-$\SPACE{\log(n)}-\Storage$-L
  $\not\subseteq $
  N-aux-$\SPACE{\log(n)}-\Storage'$-L
  then $1$D-$\Storage \not\subseteq 1$N-$\Storage'$. In particular
  $\VAL(\Storage) \notin 1$N-$\Storage'$.
\end{corollary}
\begin{proof}
  Proof by contraposition:
  If  $\VAL(\Storage) \in 1$N-$\Storage'$ then
  N-aux-$\SPACE{\log(n)}-\Storage$-L$
  = T^{-1}(\{\VAL(\Storage)\})
  \subseteq
  T^{-1}(1$N-$\Storage') = $
  N-aux-$\SPACE{\log(n)}-\Storage'$-L
  \qed
\end{proof}

Due to Lemmas \ref{lem:AltAutomataReducibleToAlternatingLogSpace} and
\ref{lem:spaceboundedTMtoReachabilityofAutomata}, complexity results
on control state reachability for $\Storage$-automata help to separate
the classes   N-aux-$\SPACE{\log(n)}-\Storage$-L for different
storage types
$\Storage$ as follows.

\begin{lemma}
  Let $\Storage$ be some storage type
  and $\mathcal{C}$ a complexity class closed under $\DSPACE(\log(n))$
  reductions. If control state reachability
  for $\Storage$-automata is complete
  $\mathcal{C}$ (under $\DSPACE(\log(n))$-reductions), then
  N-aux-$\SPACE{\log(n)}-\Storage$-L$
  = \mathcal{C}$.
\end{lemma}
\begin{proof}
  Assume that $\Aa$ is a nondeterministic  auxiliary $\SPACE{\log(n)}
  \Storage$ automaton accepting some language $L$.
  Then $L\in \mathcal{C}$ because
  Lemma \ref{lem:spaceboundedTMtoReachabilityofAutomata} provides a
  $\DSPACE(\log(n))$-reduction from $L$ to
  control state reachability for $\Storage$-automata which is in
  $\mathcal{C}$  by assumption. Thus, we conclude that
  N-aux-$\SPACE{\log(n)}-\Storage$-L$
  \subseteq \mathcal{C}$.

  Now let $L$ be some language in $\mathcal{C}$. There is a
  $\DSPACE(\log(n))$-reduction $\varphi$ such that for all words $w$,
  $\varphi(w)$ is an encoding of a nondeterministic
  $\Storage$-automaton $\Aa$ and a state $q$ such that
  $q$ is reachable in $\Aa$ if and only if $w\in L$.
  Due to Lemma
  \ref{lem:AltAutomataReducibleToAlternatingLogSpace}, there is a
  $\DSPACE(\log(n))$-reduction $\psi$  and a nondeterministic
  auxiliary $\SPACE{\log(n)} \Storage$-automaton $\Aa'$ such that
  $\Aa'$ accepts $\psi(\varphi(w))$ if and only if $q$ is reachable in
  $\Aa$ if and only if $w\in L$.
  Recall that logspace reducibility is a transitive relation because
  the $i$-th symbol of a  logspace reduction can be recomputed on the
  fly in logspace. Using the very same trick,
  we can define a nondeterministic auxiliary $\SPACE{\log(n)}
  \Storage$-automaton $\Aa''$ that, given the input $w$ simulates a run
  of $\Aa'$ on $\psi(\varphi(w))$. Hence, $\Aa''$ accepts $w$ if and
  only if $w\in L$. This shows that
  $L\in $
  N-aux-$\SPACE{\log(n)}-\Storage$-L.\qed
\end{proof}

The previous two lemmas directly imply
Proposition \ref{prop:EngelfrietLanguageReachConnection}.
As a corollary of this proposition, our results on reachability for
higher-order counters imply the
language separations stated in Corollary
\ref{cor:HOCA-Language-Separation}. Moreover, if
Proposition \ref{prop:EngelfrietLanguageReachConnection} separates
the languages  of $\Storage_1$-automata from those of
$\Storage_2$-automata, then $\VAL(\Storage_2)$ is an example language
that separates the two classes.

\begin{proof}[of Corollary \ref{cor:HOCA-Language-Separation}]
  Containments are all trivial.
  Strict containment of $L(\HOPS[(k-1)])$ in $L(\HOCSO[k])$ follows
  from the fact that we can recognise the language
  $\{a^nb^m\mid m\leq \exp_{k-1}(n)\}$ by a $\HOCSO[k]$
  (cf.~\cite{Blumensath2008}) but we cannot recognise it by a
  \HOPS[(k-1)] (cf~\cite{cawo03}).

  Recall that
  \begin{itemize}
  \item the languages of
    auxiliary $\SPACE{\log(n)} \PStorage^{k-1}(\PStorage)$ are exactly
    those in $\DTIME(\bigcup_{d\in\N}\exp_{k-1}(n^d))$ (cf. \cite{Engelfriet91}),
  \item  the languages of
    auxiliary $\SPACE{\log(n)} \PStorage^{k-1}(\CStorageZero)$ are exactly
    those in $\DSPACE(\bigcup_{d\in\N}\exp_{k-2}(n^d))$ due to
    Theorem \ref{thm:HOCSM-Reachability} and Lemmas
    \ref{lem:AltAutomataReducibleToAlternatingLogSpace} and
    \ref{lem:spaceboundedTMtoReachabilityofAutomata},
  \item the languages of
    auxiliary $\SPACE{\log(n)} \PStorage^{k-1}(\CStorage)$ are exactly
    those in $\DTIME(\bigcup_{d\in\N}\exp_{k-2}(n^d))$ due to
    Theorem \ref{thm:HOCSO-Reachability} and Lemmas
    \ref{lem:AltAutomataReducibleToAlternatingLogSpace} and
    \ref{lem:spaceboundedTMtoReachabilityofAutomata}.
  \end{itemize}
  Application of the previous corollary to the inequation
  \begin{equation*}
      \DTIME(\bigcup_{d\in\N}\exp_{k}(n^d)) \subsetneq
      \DSPACE(\bigcup_{d\in\N}\exp_{k}(n^d))
      \subsetneq \DTIME(\bigcup_{d\in\N}\exp_{k+1}(n^d))
  \end{equation*}
  yields
     $ L(\HOCSO[k]) \subsetneq L(\HOCSM[k])
      \subsetneq L(\HOPS[k]).$
\qed
\end{proof}

Correspondingly, $\VAL(\PStorage^{k+1})$ is a (collapsible)
higher-order  pushdown language of level $k+1$ recognised by a
deterministic automaton with $1$ state and no
$\varepsilon$-transitions which is not recognised by any
(collapsible) higher-order pushdown automaton of level $k$.

\section{Comparing Notions of Regularity}
\label{app:NotionsOfRegularity}
In this section, we compare the expressive power and succinctness of
different notions of regularity for sets of configurations of
$\PStorage(\CStorage)$ automata.
Recall that we introduced in Section \ref{sec:RegReach} a notion of regularity via
the encoding in binary trees. From now on we write
$\Encode$-regularity for this notion.

\subsection{$2$-Store Alternating Finite Automata}
\label{sec:2stores}

We will first compare $\Encode$-regularity with the notion of
regularity via $2$-store alternating finite automata
\cite{BouajjaniM04}. Since we introduce $\Encode$ only for
$\PStorage(\CStorage)$ configurations, we restrict our presentation
of $2$-store automata also to this setting. Nevertheless the ideas
presented here have straightforward extensions to the general setting
of $\PStorage(\PStorage)$ configurations.

\begin{definition}
  Let $\Aa'$ be a $\PStorage(\CStorage)$ automaton with state set $Q'$.
  An alternating $2$-store automaton $\Aa$ (with respect to $\Aa'$) is
  an alternating automata
  $\Aa=(Q, \rho, F, \Sigma, \Delta)$ where $Q'\subseteq Q$ is a finite
  set of states,
  $\rho:Q\to \{\exists, \forall\}$ splits $Q$ into existential and
  universal states,
  $F\subseteq Q$ the set of final states,
  $\Sigma\subseteq \{0,1,\bot\}\times A$ a set of transition labels
  such that $A$ is a \emph{finite} set of alternating finite automata
  over input
  alphabet $\{\bot\}$, and
  $\Delta\subseteq Q\times \Sigma \times Q$

  An accepting computation of $\Aa$ on a $\PStorage(\CStorage)$
  configuration is defined inductively.
  Let $\Sconfig=\Sconfig' (\tau, m)$ with $\Sconfig'\in
  (\{0,1,\bot\}\times \N)^*$, $\tau\in\{0,1,\bot\}$, and $m\in\N$,
  and let $q\in Q$ be a state.
  There is an accepting computation from $q$ on $\Sconfig$ if one of
  the following holds.
  \begin{enumerate}
  \item $\Sconfig=\varepsilon$ and $q\in F$, 
  \item Assume that  $\Sconfig\neq\varepsilon$ and that $\rho(q)=\exists$.
    there is a $q'\in Q$ and a $(\tau,\Bb)$ such that
    $(q,(\tau, \Bb), q')\in \Delta$, $\Bb$ accepts $\bot^m$, and
    there is an accepting computation from $q'$ on $\Sconfig'$.
   \item Assume that  $\Sconfig \neq \varepsilon$ and that $\rho(q)=\forall$.
    For all  $q'\in Q$ and a all $(\tau,\Bb)$ such that
    $(q,(\tau, \Bb), q')\in \Delta$, $\Bb$ accepts $\bot^m$, and
    there is an accepting computation from $q'$ on $\Sconfig'$.
  \end{enumerate}

  For $\Sconfig$ a   $\PStorage(\CStorage)$-configuration and $q\in Q$
  a state of $\Aa'$,
  we say $\Aa$ accepts  $(q, \Sconfig)$ if  there is an accepting
  computation of $\Aa$ from $q$ on $\Sconfig$.

  We call a set $C$ of configurations of a $\PStorage(\CStorage)$
  automaton
  \emph{$2$-store-regular} if there is a $2$-store automaton that
  accepts $(q,\Sconfig)$ if and only if $(q,\Sconfig)\in C$.
\end{definition}
\begin{remark}\label{rem:Deteminise-2-stores}
  It is not difficult to adapt the usual powerset construction in
  order to obtain a deterministic $2$-store automaton $\Aa'$ equivalent
  to a given alternating $2$-store automaton. By deterministic, we
  mean that for any state $q$ and any pushdown symbol
  $\tau\in\{0,1,\bot\}$ there is exactly one deterministic automaton
  $\Bb$ and one state $q'$ such that $(q, (\tau,\Bb), q')$ is a
  transition of $\Aa$. Of course this determinisation comes at the
  price of a blow-up of the state set.
\end{remark}

Note that $2$-store automata process the counter values stored in a
$\PStorage(\CStorage)$ configuration sequentially. Thus, these
automata cannot compare the values of different counters stored in the
pushdown. To the contrary, in the tree-encoding of a configuration two
adjacent counter values can be compared by just looking at the
position where the two corresponding branches split up.
Thus, we can define $\Encode$-regular sets whose members satisfy
certain restrictions with respect to the comparison of  adjacent
counter values.
This idea can be translated into a proof that there is a
$\Encode$-regular set which is not $2$-store-regular.
After giving this proof, we show that $2$-store-regular sets are
always $\Encode$-regular. These two results show that the expressive
power of $2$-store-regularity is strictly weaker than that of
$\Encode$-regularity.

\begin{proposition}
  There is a set $C$ of configurations such that $C$ is
  $\Encode$-regular but not $2$-store regular.
\end{proposition}
\begin{proof}
  Let $C=\{ (q, (\bot, m)(\bot,m)) \mid m\in \N\}$.

  $C$ is clearly $\Encode$-regular because $\Encode(C)$
  contains a tree $T$ if and only if
  there is some $m$ such that the only leaves of $T$ are
  $0^{m+1}$ and $0^m10$. It is straightforward to design a
  tree-automaton for this set of trees.

  Heading for a contradiction, assume that $C$ is accepted by some
  alternating $2$-store
  automaton $\Aa$.  There are two numbers $m_0\neq m_1$ such that
  the accepting runs of $\Aa$ on $(q, (\bot, m_0)(\bot,m_0))$ and
  $(q, (\bot, m_1)(\bot,m_1))$ use the same transitions of $\Aa$. In
  particular, both computations
  spawn the same alternating finite
  automata
  $\Aa_1, \dots, \Aa_m$  to accept $\bot^{m_0}$ or $\bot^{m_1}$,
  respectively. But then $\Aa$ also accepts
  $(q, (\bot, m_0)(\bot,m_1))\notin C$ which is a contradiction.
\end{proof}

\begin{lemma}
  Let $C$ be a $2$-store-regular set.
  Then $C$ is $\Encode$-regular.
\end{lemma}
\begin{proof}
  Let $\Aa$ be a $2$-store automaton that recognises $C$.
  As explained in Remark \ref{rem:Deteminise-2-stores}, we may assume
  that $\Aa$ is deterministic.
  Let $\Bb_1, \dots, \Bb_n$ be the deterministic finite automata
  appearing in the transition labels of $\Aa$. Assume that $\Bb$ is
  the product automaton of $\Bb_1, \dots ,\Bb_n$ and assume that the
  state sets of all $\Bb_i$ are pairwise disjoint.
  A tree-automaton accepting  $\Encode(C)$ works as follows. It
  basically simulates all the $\Bb_k$ in parallel along all
  branches. Moreover, at every branching point of the tree it guesses
  the transition of $\Aa$ that connects the element of the pushdown
  encoded in the rightmost branch of the left subtree with the
  leftmost branch of the right subtree. The precise procedure is as
  follows.

  Let $c:=(q,(\tau_1,c_1)\dots(\tau_n,c_m))$. For each node $d$ of
  $\Encode(c)$ the subtree of nodes comparable to $d$ encodes some
  subpart $(q, (\tau_i,c_i) \dots (\tau_j,c_j))$ for $1\leq i \leq j
  \leq m$. An accepting run on $\Encode(c)$ will label this node $d$
  with a tuple $(q, p,r,s)$ where $q, s$ are states of
  $\Aa$, $p$ a state of $\Bb$ and $r$ a state of some $\Bb_k$
  ($1\leq k \leq n)$
  such that
  there is a run of $\Aa$ from state $q$ to state $s$ on
  $(\tau_i,c_i)\dots(\tau_j,c_j)$ such that the first transition of
  this run spawns a copy of $\Bb_k$ along the word $\bot^{c_i}$.
  This labelling is carried out in such a way that the labels of
  different nodes are compatible in the sense that the runs witnessed
  by the labels can be composed to an accepting run of $\Aa$ on $c$.

  For this purpose, the left successor of the root is labelled by
  $L_0:=(q_0,p_0,r_0,s_0)$ where $q_0=q$, $s_0$ is a final state of
  $\Aa$, $p_0$ is the initial state of some $\Bb_k$ and $r_0$ is the
  initial state of $\Bb$. Now the states are propagated as follows:
  \begin{itemize}
  \item If a node $d$ with label $L_d=(q_d,p_d,r_d,s_d)$ has only a
    left successor (which is not a leaf, i.e., the tree label of
    $d0$ is $\bot$), then set
    $L_{d0}:=(q_d,p_{d0},r_{d0},s_d)$ such that
    $p_{d0}$ is the unique state such that $(p_d,\bot,p_{d0})$ is a
    transition of $\Bb_k$. Similarly $r_{d0}$ is the successor of
    $r_d$ with respect to $\Bb$.
  \item If a node $d$ with label $L_d=(q_d,p_d,r_d,s_d)$ has a
    left successor (which is not a leaf, i.e., the tree label of
    $d0$ is $\bot$) and a right successor, then set
    $L_{d0}:=(q_d,p_{d0},r_{d0},s_{d0})$ and
    $L_{d1}:=(s_{d0},p_{d1},r_{d},s_d)$
    such that the following holds.
    $p_{d0}$ is the unique state such that $(p_d,\bot,p_{d0})$ is a
    transition of $\Bb_k$. Similarly $r_{d0}$ is the successor of
    $r_d$ with respect to $\Bb$. $s_{d0}$ is some state of $\Aa$ and
    $p_{d1}$ is one of the components of $r_{d}$, i.e., a state of one
    of the $\Bb_{k'}$ as simulated by $\Bb$ up to this position.
  \item If the left  successor of $d$ is a leaf, and $d$'s label is
    $L_d=(q_d,p_d,r_d,s_d)$ then
    we first compute $L_{d0}$ and (if necessary) $L_{d1}$ as in the
    steps before. If the right successor exists, it is labelled by
    $L_{d1}$, the left successor is labelled by an accepting state
    if $p_d$ is an accepting  state of $\Bb_k$ such that
    $(q_d, (\tau, \Bb_k), s_{d0})$ is a transition of $\Aa$ where $\tau$
    denotes the tree-label of the leaf at $d0$.
  \end{itemize}
  It is tedious but straightforward to prove that this tree-automaton
  accepts an encoding of a configuration if and only if it is in $C$.
  By taking a product with a tree-automaton recognising only valid
  encodings of configurations the claim is proved.
\end{proof}

Unfortunately, the previous result that $\Encode$-regularity is more
expressive that $2$-store-regularity does not imply that our result on
the backwards or forward reachability carries over to
$2$-store-regular sets of configurations. The translation from the
previous proof causes a blow-up of the state spaces.
In the next lemma, we show that this blow-up is inevitable even
if we start with deterministic $2$-store automata.

\begin{lemma}\label{lem:Tree-reg-prime-test-expensive}
  There is a sequence $(\Aa_n)_{n\in\N}$ of deterministic $2$-store
  automata such that there is no polynomial $p$ and a sequence
  $(\Bb_n)_{n\in\N}$ of tree automata such that
  $\lvert \Bb_n\rvert \leq P(\lvert \Aa_n \rvert)$ and
  $\Bb_n$ accepts the same language as $\Aa_n$ (modulo translation
  with $\Encode$.
\end{lemma}
\begin{proof}
  It is easy to design a deterministic $2$-store multi-automaton $\Aa_n$
  that accepts a configuration $(\bot,v_1)(\bot,v_2)\dots(\bot,v_m)$
  if and only if
  \begin{enumerate}
  \item $m=n$, and
  \item $v_i = 0 \mod p_i$ for all $1\leq i \leq m$, where $p_i$ denotes
    the $i$-th prime.
  \end{enumerate}
  This is the automaton that goes from $q_1$ to $q_2$ to $\dots$ to
  $q_{n+1}$ spawning in the $i$-th step an automaton
  checking the length of the input modulo
  $p_i$. This automaton can be realised with $n+1+\sum\limits_{i=1}^n
  p_i\in O(n^3)$ states.

  Assume that there is a $\Bb_n$ with less than $2^n$ many states
  accepting $\Encode(C)$ for $C$ the configurations accepted by $\Aa_n$.

  Set $m:=\prod_{i=1}^n p_i$.
  There is an accepting run of $\Bb_n$ on $T:=\Encode(q,(\bot, m)(\bot,m)
  \dots (\bot,m))$ where the height of the encoded pushdown is $n$.
  Note that $m> 2^n$ and the leaves of $T$ are the nodes
  $0^{m+1}, 0^m10, 0^m1^20, \dots, 0^m1^{n-1}0$.
  Application of the pumping lemma for tree-automata yields that there
  is some $0< k < m$ and a tree whose leaves are
  $0^{k+1}, 0^k10, 0^k1^20, \dots 0^k1^{n-1}0$ accepted by $\Bb_n$.
  But this tree encodes $(q, (\bot, k)(\bot,k)\dots (\bot k))$ where
  $k$ is not divisible by all primes $p_1, \dots, p_m$.
  This contradicts the assumption that $\Bb_n$ accepts $\Encode(c)$ if
  and only if $\Aa_n$ accepts $c$.
\end{proof}

\subsection{Regularity via Sequences of Operations}
\label{sec:CarayolsRegularity}

Carayol \cite{Carayol05} introduced a notion of regularity based on
sequences of pushdown operations. He proved a normal form for this kind
of regular sets which we present in the next definition. His notion
also extends to higher-level pushdowns but for our purpose it suffices
to restrict the presentation to sets of  $\PStorage(\CStorage)$
configurations. In the following we write $\Reg(A)$ for the set of
regular expressions over alphabet $A$ and we write $L(r)$ for the
languages $L\subseteq A^*$ defined by some regular expression $r\in
\Reg(A)$.

\begin{definition}
  Let $t,s\in \Reg(\{\bot\})$ and $\sigma\in\{0,1,\bot\}$.
  Then we define a binary relation
  $\xrightarrow{test(t)\sigma s}$ on $\Sigma\times\{\bot\}^*$ by
  $(\tau, \bot^k) \xrightarrow{t\sigma s} (\tau', \bot^{k'})$
  if and only if $\bot^k\in L(t), \tau'=\sigma$, $k'\geq k$  and
  $\bot^{k'-k}\in L(s)$.
  We also define
  $(\tau, \bot^k) \xrightarrow{ s test(t) \sigma }
  (\tau', \bot^{k'})$ if and only if
  $k'\leq k$,
  $\bot^{k-k'}\in L(s)$,  $\bot^{k'}\in L(t)$ and $\tau'=\sigma$.
  These kind of definitions extend to expressions
  $e=\bigvee_{i=1}^{n_1} test(t_i)\sigma_i s_i \vee
  \bigvee_{i=1}^{n_2} r_i test(q_i) \tau_i$ via
  $\xrightarrow{e}:=
  \bigcup_{i=1}^{n_1} \xrightarrow{test(t_i)\sigma_i s_i}\cup
  \bigcup_{i=1}^{n_1} \xrightarrow{r_i test(q_i) \tau_i}$.
  We write $\bigcup_e \xrightarrow{e}$ for the set of all
  such relations.

  A \emph{sequence regular expression} is an expression
  $e= \bigvee_{i=1}^n r_is_i$ where each $r_i\in\Reg(\{\bot\})$ and
  each  $s_i\in \Reg(\bigcup_e \xrightarrow{e})$.

  Each sequence-regular expression $e=\bigvee_{i=1}^n r_is_i$ defines
  a set of configurations  $L(e)$ as follows:
  $(\sigma_0, m_0)(\sigma_1,m_1)\dots(\sigma_n,m_n)\in L(e)$ if and
  only if
  $\sigma_0=\bot$ and there is some $i\leq n$ such that
  $m_0\in L(r_i)$ and  there is a $w\in L(s_i)$ such that
  $w= \xrightarrow{e_0}\xrightarrow{e_1} \dots
  \xrightarrow{e_{n-1}}$ such that for all $j<n$
  $(\sigma_j,m_j)\xrightarrow{e_j} (\sigma_{j+1}, m_{j+1})$.

  We call a set $C$ of configurations \emph{sequence-regular} if and only
  if there is a sequence-regular expression $e$ such that $C=L(e)$.
\end{definition}

The main observation of this section is that the sets of sequence-regular
sets are a strict subset of the set of
tree-regular sets via $\Encode$. For our proof we assume the reader to
be familiar with pebble automata (cf.~\cite{Bojanczyk07} for a survey).
Moreover, we use the following results.

\begin{lemma}[\cite{Bojanczyk07}, Theorem 12 (cf.~also \cite{tenCateS2008}]
  Positive cutting caterpillar expressions define the same tree
  languages as pebble automata.
\end{lemma}

\begin{lemma}[\cite{BojanczykSSS06}, Theorem 1.1]
  The languages recognised by pebble automata are a strict subset of
  the languages recognised by tree-automata.
\end{lemma}
\begin{remark}
  We thank Miko{\l}aj Boja{\'n}czyk for pointing out that the
  separating example can
  be easily adapted to be a set of trees $T$ such that
  $T=\Encode(C)$ for a set of configurations $C$. Basically, one first
  translates the example into a set of unlabelled trees by encoding the
  labels as certain subtrees and then one adds the labels necessary to
  make the trees encodings of configurations.
\end{remark}

\begin{theorem}The following holds:
  \begin{itemize}
  \item  For each sequence-regular set $C$ of configurations,
    $C$ is $\Encode$-regular.
  \item There is an $\Encode$-regular set which is not
    sequence-regular.
  \end{itemize}
\end{theorem}
\begin{proof}
  In fact,
  if $P$ is sequence-regular, then $\Encode(P)$ is defined by a
  positive cutting caterpillar expression.
  This is due to the fact that the inorder traversal of
  $\Encode(c)$ for some configuration $c$
  visits the maximal paths in the order in which they appear as
  elements of the pushdown.

  Assume that $P$ is described by $r=\bigvee_{i=1}^n r_i s_i$.
  We translate $r$ by structural induction into a (positive cutting)
  caterpillar expression $r'$ recognising $\Encode(P)$.
  Since caterpillar expressions are closed under finite unions, it
  suffices to translate $r_i s_i$.
  Fix a $\PStorage(\CStorage)$ configuration
  $c= (\bot, m_0)(\sigma_1,m_1)(\sigma_2,m_2)\dots(\sigma_n,m_n)$ and
  let $T=\Encode(c)$.
  In order to check that $c\in L(r_i s_i)$ we first have to check that
  $m_0\in L(r_i)$. But this is equivalent to check that the leftmost
  branch of $T$ is of the form $m_0 \bot$. Thus, we modify $r_i$ to
  $r_i'$ by inserting a move to the left child before any letter
  occurring in $r_i$ and add a final move to the left child and a check
  that the leaf is labelled by $\bot$.

  Next we describe how to translate $s_i$ into a caterpillar
  expression $s_i'$ which leads to acceptance from the leftmost leaf
  of $T$ if and only if $c\in L(r_i s_i)$ and the path to this leaf
  satisfies $r_i'$.
  Recall that $s_i$ is a regular expression over relations
  $\overset{e}{\rightarrow}$. In order to satisfy $s_i$, we need to
  find a sequence of relations $\overset{e_i}{\rightarrow}$ such that
  $(\sigma_i, m_i) \overset{e_i}{\rightarrow}
  (\sigma_{i+1}, m_{i+1})$.
  Note that $(\sigma_i, m_i)$ and $(\sigma_{i+1}, m_{i+1})$ are
  encoded by the paths to $2$ adjacent leaves (in the inorder
  traversal). Thus, it suffices to gives a caterpillar expression
  $cp(e_i)$ that describes a pebble-automaton that runs from one leaf
  to the next leaf in the inorder traversal if and only if
  the corresponding paths are connected by
  $\overset{e_i}{\rightarrow}$. Once we have obtained such an
  expression, we can replace every occurrence of
  $\overset{e_i}{\rightarrow}$ by $cp(e_i)$ in $s_i$ and composition
  of the resulting expression $s_i'$ with $r_i'$ has the property that
  $r_i' s_i'$ describes a pebble-automaton run from the root to the
  rightmost leaf on $T$ if and only if $T=\Encode(c)$ for some
  $c\in L(r_i s_i)$ which completes the proof.

  In order to obtain the expression $cp(e)$ for any relation
  $\overset{e}{\rightarrow}$ we make a case distinction on the form of
  $e$.
  \begin{itemize}
  \item Assume that $e= test(t) \sigma s$. In this case $cp(e)$ first
    uses the nesting operator in order to spawn a subexpression
    $t'$ where every occurrence of $\bot$ in $t$ is replaced by
    an arbitrary sequence of moves from a right successor to its
    parent and then one move from a left successor to a $\bot$
    labelled parent. Afterwards, the main expression checks that we
    are at a leaf that is a left child, 
    we move to the parent, then to the right child and then as in the
    translation of $r_i$ we execute $s$ along the leftmost branch of
    this subtree. Additionally, we check that the leaf of this
    leftmost branch is labelled by $\sigma$.
  \item Assume that $e= s test(t) \sigma$. In this case $cp(e)$ goes
    to the parent node until coming from a left child the node has a
    right child. Then it spawns a subexpression to the left child
    which evaluates $s$ along the rightmost branch of this subtree.
    It also spawns a subexpression $t'$ as in the previous
    case. Finally it goes to the right child and then to the left
    child. There it checks that this node is a leaf labelled
    $\sigma$.
  \end{itemize}
  Finally, when $cp(e)$ reaches the rightmost leaf, it accepts the
  whole tree.

  Now using the expression $cp(e)$ instead of
  $\overset{e}{\rightarrow}$ in the $s_i$ we can translate
  $r=\bigvee_{i=1}^n r_i s_i$ into $r'=\bigvee_{i=1}^n r'_i s'_i$ and
  obtain a positive cutting caterpillar expression that recognises
  $\Encode(P)$ for $P$ the sequence-regular set we started with.
\end{proof}
\begin{remark}
  Similarly to our construction, it is easy to translate $2$-store
  automata (after determinisation) into caterpillar expressions or
  sequence-regular expression. Thus, we have a strict hierarchy with
  respect to
  expressive power from $2$-store-regularity via sequence-regularity
  to $\Encode$-regularity.
\end{remark}

As in the case of $2$-store-regularity, sequence-regularity may
provide more succinct descriptions of regular sets.
\begin{lemma}
  There is a sequence of sequence-regular expressions $r_n$ of size
  polynomial in $n$ such that there is no sequence of tree-automata
  $\Aa_n$ of size polynomial in $n$ such that $\Aa_n$ recognises
  $\Encode(L(r_n))$ for each $n\in\N$.
\end{lemma}
\begin{proof}
  Let $C$ be the set of configurations
  $(q, (\bot,c_1)(\bot,c_2)\dots(\bot, c_n))$ such
  that $c_1=c_2=\dots=c_n$ and $c_i$ is divisible by the $i$-th prime.
  It is straightforward to write down a sequence-regular expression of
  polynomial size in $n$ that describes $C$:

  $r= r s$ where
  $r=(\bot^2)^*$
  and $s=\overset{e_2}{\rightarrow}\overset{e_3}{\rightarrow} \dots
  \overset{e_{n}}{\rightarrow}$
  where $\overset{e_i}{\rightarrow}$ is the identity function on all
  $(\bot, \bot^m)$ such that $m$ is divisible by the $i$-th prime
  (which basically amounts to spawning the test $test((\bot^{p_i})^*)$).

  To the contrary, as we have seen in the proof of Lemma
  \ref{lem:Tree-reg-prime-test-expensive}, a tree-automaton
  recognising $\Encode(C)$ needs at least
  $2^n$ many states.
\end{proof}

\section{Comparison of Expressive Power}
\label{app:Storage_Equivs_Long}
In the last decades
several equivalent definitions of higher-order pushdowns were
used. Each of these can be restricted to unary stack alphabets
resulting in a priori different kinds of storage types that could be
called higher-order counters. In the following we show that most of
these variants lead to storage types that lead to equivalent notions
of nondeterministic higher-order counter automata. Note that our
definition of higher-order counters leads
to the most expressive variant of deterministic higher-order counter
automata among those that we consider in the following.

First we consider higher-order
pushdowns where only level 1 pushdowns contain stack symbols. By this
we mean that a level $k$ pushdown is not a list of pairs of stack
symbols and level $k-1$ pushdowns but only a list of level $k-1$
pushdowns. We will show that this definition is equivalent to our
definition in the case of higher-order pushdowns (which is well-known
and straightforward) as well as for higher-order counters with
$0$-test. For higher-order counters without $0$-test, we do not know
whether the two versions are equivalent. At least it is
clear that our version can simulate the more restricted version without
higher-level pushdown symbols. Thus, our upper bounds, in
particular the polynomial time algorithm for reachability on level
$2$, carry over to this setting. 
In fact, a simple adaptation of Slaat's proof \cite{Slaats12}
that  \HOCSM[k] can simulate $\HOPS[(k-1)]$
shows
that nondeterministic $\HOPS[(k-1)]$ can be simulated by 
this
restricted version of nondeterministic $\HOCSO[k]$ the lower bounds
also holds. 

Finally, we also consider higher-order pushdown automata with inverse
push-operations (cf.~\cite{cawo03,woehrle05}). In these systems
the level $k$ $\pop$-operation is replaced by a restricted version
which is only applicable if the two topmost level $k-1$ pushdowns
coincide. Carayol and Woehrle \cite{cawo03}  have shown that
this  kind of higher-order pushdown storage is equivalent to the usual
one for nondeterministic automata
(see \cite{woehrle05} for a proof). We show that this carries over to
higher-order counters. In
particular, even when we replace $\pop$ by the inverse push-operation
and do not allow pushdown symbols on higher levels, the resulting
nondeterministic higher-order counter automata with $0$-test (without
$0$-test, respectively)
is still equivalent
to our notion of nondeterministic higher-order counter automata with
$0$-test (without $0$-test, respectively). Some of these results carry
over to deterministic automata as well. Before we go into the details
we summarise our results in
the following two theorems. We say $\Storage$ automata simulate $\Storage'$
automata if for each $\Storage'$ automaton there is a $\Storage$
automata recognising the same language and generating the same
configuration graph after $\varepsilon$-contraction.

\begin{theorem} \label{thm:equivStoragesNondeterministic}
  The following holds:
  \begin{enumerate}
  \item
    For any of the following storage types  the
    nondeterministic $r$-way automata of one type can simulate the
    nondeterministic $r$-way automata of another type for ($r\in\{1,2\}$):
    \begin{itemize}
    \item $\PStorage^k(\PStorage)$, i.e., level $k$ pushdown automata with
      pushdown symbols on each level (Engelfriet's definition of level
      $k$ pushdown automata),
    \item $\CStorage^k(\PStorage)$, i.e., level $k$ pushdown automata
      with pushdown symbols only on level $1$ (used for instance in
      \cite{HagueOng08}),
    \item $\CStorage^k_{inv}(\PStorage)$,
      i.e., level $k$ pushdown automata
      with pushdown symbols only on level $1$ and inverse push
      operations (introduced in
      \cite{cawo03}),
    \item $\PStorage^k_{inv}(\PStorage)$,
      i.e., level $k$ pushdown automata
      with pushdown symbols on each level and  with inverse push
      operations.
    \end{itemize}
  \item
    The analogous statement for nondeterministic higher-order counter
    automata with $0$-test also holds.
    For any of the following storage types  the
    nondeterministic $r$-way automata of one type can simulate the
    nondeterministic $r$-way automata of another type (for
    $r\in\{1,2\}$):
    \begin{itemize}
    \item $\PStorage^k(\CStorageZero)$, i.e., level $k$ counter
      automata (with $0$-test) with
      pushdown symbols on each level,
    \item $\CStorage^k(\CStorageZero)$, i.e., level $k$ counter
      automata (with $0$-test)
      with pushdown symbols only on level $1$,
    \item $\CStorage^k_{inv}(\CStorageZero)$,
      i.e., level $k$ counter
      automata (with $0$-test)    with pushdown symbols only on level
      $1$ and inverse push  operations,
    \item $\PStorage^k_{inv}(\CStorageZero)$,
      i.e., level $k$ counter
      automata (with $0$-test) with pushdown symbols on each level and
      with inverse push  operations.
    \end{itemize}
  \item
    For nondeterministic higher-order counter
    automata without $0$-test we can only prove a weaker result.
    For any of the following storage types  the
    nondeterministic r-way automata of one type can simulate the
    nondeterministic r-way automata of another type (for $r\in\{1,2\}$):
    \begin{itemize}
    \item $\PStorage^k(\CStorage)$, i.e., level $k$ counter
      automata (without $0$-test) with
      pushdown symbols on each level,
    \item $\CStorage^k_{inv}(\CStorage)$,
      i.e., level $k$ counter
      automata (without $0$-test)    with pushdown symbols only on level
      $1$ and inverse push  operations,
    \item $\PStorage^k_{inv}(\CStorage)$,
      i.e., level $k$ counter
      automata (without $0$-test) with pushdown symbols on each level and
      with inverse push  operations.
    \end{itemize}
    Moreover, nondeterministic
    $\CStorage^k(\CStorage)$ automata can be simulated by
    each of the above mentioned automata.
  \end{enumerate}
\end{theorem}
\begin{remark}
  All results of this theorem carry over to alternating automata
  analogously. Moreover, we can also add an auxiliary tape of
  size $\SPACE{b(n)}$ for arbitrary function $b$.
\end{remark}

\begin{theorem}\label{thm:equivStoragesDeterministic}
  The following holds:
  \begin{enumerate}
  \item
    For any of the following storage types the
    deterministic $r$-way automata of one type can simulate the
    deterministic $r$-way automata of another type (for $r\in\{1,2\}$):
    \begin{itemize}
    \item $\PStorage^k(\PStorage)$, i.e., level $k$ pushdown automata with
      pushdown symbols on each level,
    \item $\CStorage^k(\PStorage)$, i.e., level $k$ pushdown automata
      with pushdown symbols only on level $1$,
    \item $\PStorage^k_{inv}(\PStorage)$,
      i.e., level $k$ pushdown automata
      with pushdown symbols on each level and  with inverse push
      operations.
    \end{itemize}
    Moreover deterministic
    $\CStorage^k_{inv}(\PStorage)$ automata can be simulated by any of
    the above mentioned automata types.
  \item
    The analogous statement for deterministic higher-order counter
    automata with $0$-test also holds.
    For any of the following storage types  the
    deterministic $r$-way automata of one type can simulate the
    deterministic $r$-way automata of another type (for $r\in\{1,2\}$):
    \begin{itemize}
    \item $\PStorage^k(\CStorageZero)$, i.e., level $k$ counter
      automata (with $0$-test) with
      pushdown symbols on each level,
    \item $\CStorage^k(\CStorageZero)$, i.e., level $k$ counter
      automata (with $0$-test)
      with pushdown symbols only on level $1$,
    \item $\PStorage^k_{inv}(\CStorageZero)$,
      i.e., level $k$ counter
      automata (with $0$-test) with pushdown symbols on each level and
      with inverse push  operations.
    \end{itemize}
    Moreover deterministic
    $\CStorage^k_{inv}(\CStorageZero)$ automata can be simulated by any of
    the above mentioned automata types.
  \item
    Deterministic $\PStorage^k(\CStorage)$ automata
    can simulate deterministic $\PStorage^k_{inv}(\CStorage)$ automata
    and vice versa. Moreover,
    deterministic $\CStorage^k_{inv}(\CStorage)$ automata and
    deterministic $\CStorage^k(\CStorage)$ automata are strictly
    weaker that $\PStorage^k(\CStorage)$ automata in the sense that
    every automaton of one of the former types can be simulated by
    some automaton of the
    latter type but not vice versa.
  \end{enumerate}
\end{theorem}
\begin{remark}
  The proof of this theorem will be based on Engelfriet's notion of
  equivalent storages. Thus, the statement remains valid, if
  we replace deterministic automata by
  any other kind of deterministic/nondeterministic/alternating $r$-way
  auxiliary $\SPACE{b(n)}$ automata. It even carries over to the
  corresponding classes of transducers.
\end{remark}

We conclude the presentation of the results of this section by
pointing the reader to the open problems concerning equivalence of
storage types.
\begin{problem}
  \begin{enumerate}
  \item Is there some nondeterministic $\PStorage^k(\CStorage)$
    automaton that cannot be simulated by any nondeterministic
    $\CStorage^k(\CStorage)$ automaton?
  \item Can we determinise the storage simulations that we so far only
    realised nondeterministic? In other words, can
    deterministic $\CStorage^k(\Storage)$ automata be simulated by
    deterministic $\CStorage^k_{inv}(\Storage)$ automata for
    $\Storage$ one storage type of the set $\{\PStorage,
    \CStorageZero, \CStorage\}$?
  \end{enumerate}
\end{problem}

\subsection{Simulation of Deterministic Automata}
We first prove our claims about the deterministic case.
Note that the nontrivial claims of Theorem
\ref{thm:equivStoragesDeterministic} will be  proved in Propositions
\ref{prop:equivStorages1} and \ref{prop:equivStorages2}, and in
Corollary \ref{cor:WeakerStoragesDeterministic}.
Let us first  recall the notion of equivalence of storage types
(cf.~\cite{Engelfriet91}).

\begin{definition}
  Let $\Storage$ and $\Storage'$ be storages.
  $\Storage$ can \emph{simulate}
  $\Storage'$, denoted as $\Storage' \preceq \Storage$,
  if
  for every
  one-way deterministic $\Storage'$ transducer there is a one-way
  deterministic $\Storage$ transducer defining the same transductions.

  $\Storage$ and
  $\Storage'$ are \emph{equivalent},
  denoted as $\Storage \equiv \Storage'$, if
  $\Storage \preceq \Storage'$ and $\Storage'\preceq \Storage$.
\end{definition}
\begin{remark}
  As pointed out by Engelfriet, this notion  of equivalence implies
  that if $\Storage\preceq\Storage'$, then  for
  $t\in\{\text{nondeterministic, alternating, deterministic}\},
  r\in\{1,2\}$, the $t$ $r$-way auxiliary $\SPACE{b(n)}$ \Storage
  automata can be simulated by the
  $t$ $r$-way auxiliary $\SPACE{b(n)}$ \Storage'
  automata
\end{remark}

Recall that we defined the storage type
$\CStorage = \PStorage_{\{\bot\}}$. In the following, we also use
$\CStorage$ as the operator $\PStorage_{\{\bot\}}$ acting on other storage
types. We call $\CStorage(\Storage)$ the storage type \emph{counter of
  $\Storage$}.
Apparently, the (first) $\bot$ component of every entry in the elements of
$\CStorage(\Storage)$ is redundant. Identifying
$(\bot, c_1) (\bot,c_2) \dots
(\bot, c_n)$ with $(c_1)(c_2)\dots(c_n)$ one sees easily that
$\CStorage^{k-1}(\PStorage)$ automata are equivalent to the higher-order
pushdown automata variant (of level $k$) used for instance in
\cite{HagueOng08}.

\begin{proposition}\label{prop:equivStorages1}
  $\CStorage^k \preceq \PStorage^{k-1}(\CStorage)$,
  $\CStorage^{k-1}(\CStorageZero) \equiv
  \PStorage^{k-1}(\CStorageZero)$  and
  $\CStorage^{k-1}(\PStorage) \equiv
  \PStorage^{k-1}(\PStorage)$.
\end{proposition}
\begin{proof}
  The direction from left to right is clear because $\PStorage$ is an
  extension of $\CStorage$.
  We next  show how $\CStorage^{k-1}(\CStorageZero)$ can simulate
  $\PStorage^{k-1}(\CStorageZero)$.

  We first show that $\Storage:=\PStorage(\CStorage^{k-2}(\CStorageZero))$
  can be simulated by $\Storage':=\CStorage^{k-1}(\CStorageZero)$. The
  idea is to
  encode the pushdown symbol of level $k$, by the level 1 counter
  value modulo $3$ (recall that $\PStorage$ uses the  pushdown
  alphabet $\{\bot, 0, 1\}$).
  For this purpose we first replace in the
  $\PStorage(\CStorage^{k-2}(\CStorageZero))$ every $\push{\bot}$ of
  level $1$(i.e. a push applied to $\CStorageZero$) by $3$
  $\push{\bot}$ operations and each level $1$ $\pop$-operation by $3$
  $\pop$-operations of level $1$.
  This results in an equivalent $\Storage$ automaton where the level
  $1$ counter value is always $0 \mod 3$. Next, without loss of
  generality we assume that the $\Storage$ automaton only uses
  instructions of the form $\pop, \push{\sigma,\id}$ and $\stay{f}$.
  \emph{For the rest of this simulation, we identify $\bot$ with the
    number $2$}. We want to represent a
  pushdown symbol $\sigma\in\{0,1,\bot\}$ by $\sigma \mod 3$ on the
  level $1$ counter.
  We initialise $\Storage'$ by applying $2$ $\push{\bot}$ on level
  $1$ (this results in the counter value $2$, which is  $2\mod 3$
  representing the initial symbol $\bot$. Now we simulate the
  operations on $\Storage$ by $\Storage'$-operations as follows (where
  we assume that the current $\Storage$-configuration $\Sconfig$ is
  simulated by $\Storage'$-configuration $\Sconfig'$.
  \begin{enumerate}
  \item The $\TOP{\gamma}$ test for $\gamma\in\{0,1,\bot\}$ can  be
    simulated as follows.
    apply $\push{\bot,\id}$, then
    determine the topmost symbol
    $\gamma'\in\{0,1,\bot\}$
    by level $1$
    $\pop$-operations (while the $0$-test fails) determining the value
    of the topmost level $1$
    counter modulo $3$. After finishing the test we restore the
    pushdown by a $\pop$ operation and just have to compare $\gamma$
    with $\gamma'$.
  \item The $\Temptystack$ test on level $1$ is simulated by first
    determining which $\TOP{\gamma}$ test applies for
    $\gamma\in\{0,1,\bot\}$ as in the simulation of
    $\TOP{\gamma}$. Then we perform
    $\gamma$ many
    $\pop$-operations of level $1$, then the $\Temptystack$ test of
    $\Storage'$  coincides with the $\Temptystack$ of $\Storage$. We
    restore the pushdown by $\gamma$ many $\push{\bot}$ operations of
    level $1$.
  \item A $\pop$ operation is simulated by $\pop$.
  \item A $\push{\gamma,\id}$ operation is simulated by the following
    program:
    first determine the topmost symbol $\gamma'\in\{0,1,\bot\}$
    of $\Sconfig'$. Then apply $\push{\bot,\id}$, then apply $\gamma'$
    many level $1$ pop-operations.  No we apply $\gamma$ many level
    $1$ $\push{\bot}$ operations.
  \item A $\stay{f}$ operations is simulated by $\stay{f}$ if $f$ is
    not an operation of level $1$. If it is of level $1$ we just
    duplicate it $3$ times.
  \end{enumerate}

  This completes the proof that
  $\PStorage(\CStorage^{k-2}(\CStorageZero))$ can be simulated by
  $\CStorage^{k-1}(\CStorageZero)$. The lemma now follows
  by induction on $k$:
  we have shown that
  $\PStorage^1(\CStorageZero) \equiv \CStorage^1(\CStorageZero)$.
  Assume that for some $k$ we have $\PStorage^{k-1}(\CStorageZero)
  \equiv \CStorage^{k-1}(\CStorageZero)$.
  By Theorem 1.3.1  of \cite{Engelfriet91},
  we obtain
  \begin{equation*}
    \PStorage^k(\CStorageZero) \equiv
    \PStorage( \PStorage^{k-1}(\CStorageZero)) \equiv
    \PStorage( \CStorage^{k-1}(\CStorageZero)) \equiv
    \CStorage^k(\CStorageZero).
  \end{equation*}
  The equivalence
  $\PStorage^{k}(\PStorage) \equiv
  \CStorage^k(\PStorage)$ is obtained completely analogous.
\end{proof}

We now want to discuss the variants of pushdown systems and counters
with inverse push-operations.
\emph{For reasons of simplicity, we now consider the operator
  $\PStorage$ to be restricted to $\push{\sigma,\id}$, $\stay{f}$ and
  $\pop$ operations.}
Let $\PStorage_{inv}$ and $\CStorage_{inv}$ be the
variants of (the restricted) $\PStorage$ and $\CStorage$  with inverse
push-operations,
i.e., $\PStorage_{inv}$ is defined as $\PStorage$ but instead of the
operation $\pop$ we have the operation $\invpush{\gamma,\id}$.
For $\Storage$ a storage type and
$s:=(\sigma_1, \Sconfig_1) \dots (\sigma_{m-1},
\Sconfig_{m-1})(\sigma_m,\Sconfig_m)$ a $\PStorage(\Storage)$
configuration $\invpush{\gamma,\id}(s)$ is defined if and only if
$\sigma_m=\gamma$ and  $\Sconfig_m=\Sconfig_{m-1}$, i.e., if and only
if
$\push{\gamma, \id}((\sigma_1, \Sconfig_1) \dots (\sigma_{m-1},
\Sconfig_{m-1})) = s$. In this case,
$\invpush{\gamma,\id}(s) = \pop(s) = (\sigma_1, \Sconfig_1) \dots
(\sigma_{m-1},\Sconfig_{m-1})$.

Carayol and Woehrle\cite{cawo03} already showed that
nondeterministic $\CStorage^k_{inv}(\PStorage)$ automata can
simulate nondeterministic $\CStorage^k(\PStorage)$ automata and that
$\CStorage^k_{inv}(\PStorage) \preceq \CStorage^k(\PStorage)$.
The latter simulation uses the fact that for every
$\CStorage^k_{inv}(\PStorage)$-configuration there is a unique
shortest sequence of operations that generates this configuration from
the initial one. Moreover, a sequence $s$ of operations  translates one
configuration $\Sconfig_1$into another configuration $\Sconfig_2$ if
and only if the following holds. Let $s_i$ be the unique sequence
generating $\Sconfig_i$, then $s_2$ results from $s_1s$ by removing
all adjacent pairs of inverse operations. Here, the inverse of
$\push{\gamma, \id}$ is $\invpush{\gamma}$ and the inverse of level
$1$ $\push{\sigma}$ is $\pop_\sigma$ and the inverse of $\stay{f}$ is
$\stay{f^{-1}}$ where $f^{-1}$ is the inverse of $f$ ($\pop_\sigma$
denotes a $\pop$ operation that is applied to a pushdown with topmost
symbol $\sigma$). We next prove a similar result for
$\PStorage_{inv}^k(\Storage)$ and $\PStorage^k(\Storage)$ for
$\Storage\in\{\PStorage, \CStorage, \CStorageZero\}$ that even work
deterministically in both directions.

\begin{lemma}\label{lem:CWlikeProof}
  For all $k\in\N$ $\PStorage_{inv}( \PStorage^k_{inv}( \Storage)) \equiv
  \PStorage( \PStorage^k_{inv}(\Storage))$ for  $\Storage\in\{\PStorage,
  \CStorage, \CStorageZero\}$.
\end{lemma}
\begin{proof}
  We first show
  $\PStorage_{inv}( \PStorage^k_{inv}( \Storage)) \preceq
  \PStorage( \PStorage^k_{inv}(\Storage))$.
  This proof adapts the one of \cite{woehrle05} and uses the
  level $k$ symbols on
  the pushdown to store the minimal sequence that generated the
  current pushdown. For this purpose we replace the operations on
  $\PStorage_{inv}(\PStorage^k_{inv}(\Storage))$ as follows.
  \begin{enumerate}
  \item $\push{\gamma,\id}$ is replaced by
    $\push{(\gamma,\invpush{\gamma}), \id}$.
  \item $\stay{f}$ applied to a pushdown $p$ represented by the
    pushdown $p'$  is replaced by $\pop$ if the test
    $\TOP{(\gamma,\stay{f})}(p')=true$ for some $\gamma\in\Gamma$.
    Otherwise, it is replaced by
    $\push{(\gamma, \stay{f^{-1}}), f}$ for $\gamma$ such that
    $\TOP{\gamma}(p)=true$.
  \item $\invpush{\gamma}$ is replaced by $\pop$ if
    $\TOP{(\gamma,\invpush{\gamma})}(p')=true$, otherwise it is
    undefined on $p$ whence the simulation stops.
  \end{enumerate}
  Adding some coding, one can translate the resulting system into one
  with topmost pushdown alphabet $\{\bot,0,1\}$.
  Correctness of this simulation follows from the results in
  \cite{cawo03,woehrle05}.

  For the other direction Carayol and Woehrle \cite{cawo03} proposed
  to simulate
  $\pop$ by guessing and creating the right level $k-1$ pushdown by
  push- and inverse push-operations of level below $k$ and then apply
  an inverse level $(k-1)$-operation. This of course is a
  nondeterministic behaviour. Instead, we use their idea from the
  translation in the other direction: we annotate the pushdowns with
  the necessary operations in order to obtain the topmost pushdown of
  level $k-1$ for which the inverse push is applicable.
  \begin{enumerate}
  \item $\push{\gamma,\id}$ is replaced by
    $\push{(\gamma,\pop_\gamma), \id}$.
  \item $\stay{f}$ applied to a pushdown $p$ represented by the
    pushdown $p'$  is replaced by
    $\stay{f}; \invpush{\gamma, \stay{f}}$ if the test
    $\TOP{(\gamma,\stay{f})}(p')=true$ for some $\gamma\in\Gamma$.
    Otherwise, it is replaced by
    $\push{(\gamma, \stay{f^{-1}}), f}$ for $\gamma$ such that
    $\TOP{\gamma}(p)=true$.
  \item $\pop$ is replaced by
    a sequence performing
    $\stay{f}; \invpush{\gamma,\stay{f}}$ while the topmost level $k$
    symbol is $(\gamma, \stay{f})$. After iteration of this
    instruction, we end up with a topmost symbol $(\gamma,\pop)$ for
    some symbol $\gamma$. We then apply $\invpush{(\gamma,\pop)}$.
  \end{enumerate}
  Again using the usual coding trick, we can restrict the level $k$
  pushdown alphabet to $\{\bot,0,1\}$.
  The proof that this simulation is correct is completely analogous to
  the proof of the other direction.
\end{proof}

This lemma allows to prove the following proposition:

\begin{proposition}\label{prop:equivStorages2}
  $\PStorage_{inv}^k(\PStorage) \equiv \PStorage^k(\PStorage)$,
  $\PStorage_{inv}^k(\CStorageZero) \equiv \PStorage^k(\CStorageZero)$,and
  $\PStorage_{inv}^k(\CStorage) \equiv \PStorage^k(\CStorage)$,
\end{proposition}
\begin{proof}
  Let $\Storage\in\{\PStorage, \CStorageZero, \CStorage\}$.
  By induction on the previous lemma and the fact that the operator
  $\PStorage$ preserves equivalence of storages
 (cf.~\cite{Engelfriet91}), we obtain
  \begin{equation*}
    \PStorage^k(\Storage)
    \equiv \PStorage(\PStorage_{inv}^{k-1}(\Storage))
    \equiv \PStorage_{inv}(\PStorage^{k-1}_{inv}(\Storage)))
    = \PStorage^k_{inv}(\Storage).
  \end{equation*}
\end{proof}

We conclude this section by showing that $\CStorage^k(\CStorage)$ is
strictly weaker than $\PStorage^k(\CStorage)$ (and analogously for the
variants with inverse push). In fact, we prove the stronger claim that
any storage type with only trivial tests cannot deterministically
simulate
$\CStorageZero$.

\begin{definition}
  Let $\Storage=(\SConfig, T, F, \SConfigInit)$ be a storage type. We
  call it \emph{test-free} if the result of each test $t\in T$ is
  independent of the tested configuration, i.e., for all
  $t\in T$, and all $\Sconfig, \Sconfig'\in\SConfig$  we have
  $t(\Sconfig)=t(\Sconfig')$.
\end{definition}
\begin{example}
  $\CStorage^k(\CStorage)$ and $\CStorage^k_{inv}(\CStorage)$ are
  test-free whereas  $\PStorage$ is not test-free.
\end{example}
In the following, we show that test-free storage types cannot
deterministically compute
any unbounded function $f$ in the sense that the language
$L_f:=\set{a^nb^{f(n)}}{n\in\N}$ is not recognised by any
deterministic $\Storage$ automaton where  $\Storage$ is a test-free
storage type.  In particular, test-free deterministic $\Storage$
automata do not accept $\set{a^nb^n}{n\in\N}$ whence
$\CStorageZero\not\preceq \Storage$.
The crucial observation is that the storage configuration has no
influence on the next transition except for the fact that it can abort
a computation.
\begin{lemma}
  Let $\Storage$ be a test-free storage type and $\Aa$ a
  deterministic $\Storage$ automaton. For each input letter $\sigma$
  and all states $q$,
  the set of storage-configurations  $\SConfig$ splits into two
  disjoint sets $\SConfig=\SConfig_b\sqcup \SConfig_t$ such that
  \begin{itemize}
  \item for
    all configurations $(q, \Sconfig)$ with $\Sconfig\in \SConfig_b$ no
    transition is
    applicable to $(q, \Sconfig)$, and
  \item there is a unique state $p$ and a unique \Storage-operation $o$
    such that the unique successor configuration on reading
    $\varepsilon$ or $\sigma$ for each
    $(q, \Sconfig)$ with $\Sconfig\in\SConfig_t$ is $(p, o(\Sconfig))$.
  \end{itemize}
\end{lemma}

By induction on the length of a run we obtain the following corollary.
\begin{corollary}
  Let $\Storage$ be a test-free storage type and $\Aa$ a
  deterministic $\Storage$ automaton.
  For each state $q$ there is a unique state $p$ and a
  $\Storage$-operation $o$ such that for each configuration
  $(q, \Sconfig)$ that admits a run on $\sigma^k$, the unique
  successor configuration after reading $\sigma$ or $\varepsilon$ is
  $(p, o(\Sconfig))$. In particular, if $(q, \Sconfig)$ and
  $(q, \Sconfig')$ both allow a run reading $\sigma^k$ these runs both
  end in the same state $p'$.
\end{corollary}

\begin{proposition}
  Let $f:\N\to\N$ be an unbounded
  function and $\Storage$ a test-free storage
  type. $L_f=\set{a^nb^{f(n)}}{n\in\N}$ is not recognised by any
  deterministic $\Storage$ automaton.
\end{proposition}
\begin{proof}
  Since $f$ is unbounded, there is a state $q$ and numbers
  $n_1,n_2\in\N$ with $f(n_1) < f(n_2)$ such that the run on
  $a^{n_i}$ ends in $(q, \Sconfig_i)$ for storage configurations
  $\Sconfig_1, \Sconfig_2$ of $\Storage$.
  By assumption there is a run from $(q, \Sconfig_1)$ reading
  $b^{f(n_1)}$ and ending in an accepting state $p$.
  Since $(q, \Sconfig_2)$ admits a run reading $b^{f(n_2)}$, it also a
  admits a run reading $b^{f(n_1)}$. Due to the previous corollary,
  this run ends in state $p$, whence $a^{n_2}b^{f(n_1)}$ is
  accepted. But this contradicts the fact that
  $a^{n_2}b^{f(n_1)}\notin L_f$ because $f(n_1)<f(n_2)$.
\end{proof}

\begin{corollary}\label{cor:WeakerStoragesDeterministic}
  $\CStorageZero\not\preceq \Storage$ for any test-free storage type
  $\Storage$. In particular,
  $\CStorageZero \not\preceq \CStorage^k(\CStorage)$ and
  $\CStorageZero \not\preceq \CStorage_{inv}^k(\CStorage)$ for all
  $k\in\N$.
\end{corollary}
\begin{proof}
  There is a deterministic $\CStorageZero$ automaton recognising
  $L_{\id}=\set{a^nb^n}{n\in\N}$ which (by the previous proposition) is
  not recognised by any deterministic \Storage automaton $\Aa$.
\end{proof}

Since obviously $\CStorageZero\preceq \PStorage^k(\CStorage)$ for all
$k\geq 1$,  $\PStorage^k(\CStorage)$ is not equivalent to
$\CStorage^k(\CStorage)$ or
$\CStorage^k_{inv}(\CStorage)$.

\begin{corollary}
  $\CStorage^k(\CStorage) \prec \PStorage^k(\CStorage)$ and
  $\CStorage^k_{inv}(\CStorage) \prec \PStorage^k(\CStorage)$ for all
  $k\geq 1$.
\end{corollary}

\subsection{Simulation of nondeterministic automata}
We now define a 'nondeterministic' version of the notion of
equivalence of storage types. This allows to prove those parts of
Theorem \ref{thm:equivStoragesNondeterministic} that are not already
implied by the results from the previous section.

\begin{definition}
  Let $\Storage=(\SConfig,T,F,\SConfigInit)$ and
  $\Storage'=(\SConfig',T',F',\SConfigInit')$ be storage types.
  We say $\Storage$ can be nondeterministically simulated by
  $\Storage'$ and write $\Storage\preceq_N \Storage'$ if there is a
  map $\varphi:\SConfig \to \SConfig'$ such that the following holds.
  \begin{enumerate}
  \item There is a sequence $f_1, f_2, \dots, f_n\in F'$ such that
    $\varphi(\SConfigInit) = f_1(f_2(\dots f_n(\SConfigInit')\dots))$.
  \item For each $f\in F$ there is a nondeterministic $\Storage'$
    automaton $\Aa_f$ with initial state $q_i$ and final state $q_f$
    such that for all $\Sconfig\in\SConfig$ there is a run of $\Aa_f$
    from $(q_i,\varphi(\Sconfig))$ to $(q_f, y')$ if and only if
    $f(\Sconfig)$ is defined and $y'=\varphi(f(\Sconfig))$.
  \item For each $t\in T$ there are two nondeterministic $\Storage'$
    automaton $\Aa_t, \bar\Aa_t$ with initial states $q_i$ and $\bar
    q_i$, and
    final states $q_f$ and $\bar q_f$, respectively,
    such that
    \begin{itemize}
    \item for all $\Sconfig\in\SConfig$ there is a run of $\Aa_t$
      from $(q_i,\varphi(\Sconfig))$ to $(q_f, y')$ if and only if
      $t(\Sconfig)=true$  and $y'=\varphi(\Sconfig)$, and
    \item for all $\Sconfig\in\SConfig$ there is a run of $\bar\Aa_t$
      from $(\bar q_i,\varphi(\Sconfig))$ to $(\bar q_f, y')$ if and only if
      $t(\Sconfig)=false$  and $y'=\varphi(\Sconfig)$.
    \end{itemize}
  \end{enumerate}
  As in the case of $\preceq$, we write $\Storage\equiv_N\Storage'$ if
  $\Storage \preceq_N \Storage'$ and $\Storage'\preceq_N \Storage$.
\end{definition}

\begin{proposition}
  Let $\Storage \preceq_N \Storage'$, $b(n):\N\to\N$, $r\in\{1,2\}$
  and $t\in\{\text{nondeterministic, alternating}\}$.
  Every $r$-way $t$ auxiliary $\SPACE{b(n)}$ \Storage automaton $\Aa$ is
  simulated by some
  $r$-way $t$ auxiliary $\SPACE{b(n)}$ \Storage' automaton $\Aa'$ in
  the sense that the configuration graphs of $\Aa$ and $\Aa'$ coincide
  after $\varepsilon$-contraction and both automata accept the same
  language.
\end{proposition}
\begin{proof}
  By a straightforward product construction of $\Aa$ and the
  $(\Aa_f)_{f\in F}$, $(\Aa_t)_{t\in T}$ and $(\bar\Aa_t)_{t\in T}$.
  Instead of executing $\Storage$-tests or -operations the automaton
  guesses the correct test result and then checks its guess and
  simulates the operation by executing first the corresponding
  $\Aa_t/\bar\Aa_t$ and the the corresponding $\Aa_f$ from its initial
  to its final state.
\end{proof}

As in the deterministic case, the pushdown operator is monotone with
respect to $\preceq_N$.
\begin{proposition} \label{prop:prececNondCompatiblePStorage}
  Let $\Storage \preceq_N \Storage'$. We have
  $\PStorage(\Storage)\preceq_N \PStorage(\Storage')$,
  $\PStorage_{inv}(\Storage)\preceq_N \PStorage_{inv}(\Storage')$,
  $\CStorage(\Storage)\preceq_N \CStorage(\Storage')$,
  $\CStorage_{inv}(\Storage)\preceq_N \CStorage_{inv}(\Storage')$,
\end{proposition}
\begin{proof}
  It suffices to provide simulations of the test $test(t)$ for each
  test $t$ of $\Storage$ and simulations for the operations
  $\stay{f}$ (note that $\push{\gamma,f}$ can be replaced by
  $\push{\gamma,\id}; \stay{f}$).

  The automaton $\Aa_{test(t)}$ that checks that $test(t)=true$ is
  equal to $\Aa_t$ but executes test $test(t')$ whenever $\Aa_t$
  executes $t'$ and performs $\stay{f'}$ whenever $\Aa_t$ performs
  $\Storage'$-operation $f'$. Analogously we define
  $\bar\Aa_{test(t)}$ and $\Aa_{\stay{f}}$.
\end{proof}

Note that all storage types $\Storage=(\SConfig, T, F, \SConfigInit)$
we consider are \emph{strongly connected} in the
sense that for any $\Sconfig, \Sconfig' \in\SConfig$ there is a
sequence $f_1, f_2, \dots, f_n$ of $\Storage$-operations such that
$\Sconfig'=f_1(f_2(\dots f_n(\Sconfig)\dots))$.
As Carayol and Woehrle already noticed, $\pop$  of
$\PStorage(\Storage)$ can be  simulated nondeterministically by
inverse push of $\PStorage_{inv}(\Storage)$ if $\Storage$ is strongly
connected by simply guessing the right $\Storage$ configuration and
restoring it before simulating the $\pop$ by an inverse push.
\begin{lemma}
  For strongly connected storage types $\Storage$,
  $\CStorage(\Storage) \preceq_N \CStorage_{inv}(\Storage)$.
  Moreover, these storage types are again strongly connected.
\end{lemma}

Induction on the previous lemma directly yields the following
proposition.
\begin{proposition}
  For $k\in\N$ and $\Storage\in\{\PStorage, \CStorageZero,
  \CStorage\}$ we have
  $\CStorage^k(\Storage) \preceq_N \CStorage^k_{inv}(\Storage)$.
\end{proposition}
\begin{proof}
  Inductively,
  $\CStorage^k(\Storage) \preceq_N
  \CStorage(\CStorage^{k-1}_{inv}(\Storage)) \preceq_N
  \CStorage_{inv}(\CStorage^{k-1}_{inv}(\Storage)) =
  \CStorage^k_{inv}(\Storage)$.
\end{proof}

The last claim we have to prove  is that
$\CStorage_{inv}^k(\CStorage) \equiv_N \PStorage_{inv}^k(\CStorage)$.
Again we first prepare the proof by induction with a simple lemma.

\begin{lemma}\label{lem:PinvtoCinvNond}
  $\PStorage_{inv}(\CStorage_{inv}^{k-1}(\CStorage)) \equiv_N
  \CStorage_{inv}(\CStorage_{inv}^{k-1}(\CStorage))$ for all $k\geq 1$.
\end{lemma}
\begin{proof}
  The $\succeq_N$ direction is trivial.
  For the other direction, we first do the case $k=1$ and then the
  case $k\geq 2$.

  A $\PStorage_{inv}(\CStorage)$ configuration
  $(\sigma_1, m_1) \dots (\sigma_n, m_n)$ is identified with the
  $\CStorage_{inv}(\CStorage)$ configuration
  $(\bot, m_1+\sigma_1) (\bot, m_1) (\bot, m_2+\sigma_2)(\bot, m_2)
  \dots (\bot, m_n+\sigma_n) (\bot, m_n)$ (where we again identify
  $\bot$ with $2$).
  In this representation a test for the topmost symbol is simple:
  the topmost symbol is $\sigma$ if we can apply $\stay{\push{\bot}}$
  $\sigma$ many times follows by inverse push, push and $\sigma$ many
  $\stay{\pop}$ operations. The corresponding negative test is by
  guessing the symbol $\tau\in\{\bot,0,1\}\setminus\{\sigma\}$ and
  applying the positive test for $\tau$.
  With the ability to test for the encoded topmost symbol, it is then
  easy to simulate any of the $\PStorage_{inv}(\CStorage)$
  operations.

  For the case $k\geq 3$ we use basically the same idea but we have to
  take care that we only encode one topmost symbol in the topmost
  level $k-1$ counter.
  For this purpose we define an auxiliary notation let
  $\sigma\in\{0,1,\bot\}$, and $m$ a
  $\CStorage_{inv}^{k-1}(\CStorage)$ configuration.
  We write $m+\sigma$ for the result of applying to $m$ the level $2$
  operation $\push{\bot,\id}$ followed by the level $1$ $\push{\bot}$
  for $\sigma$  many times (level $n$  means that we put the
  mentioned operation into a $(k-n)$-fold application of $\stay{}$).
  We then encode a
  $\PStorage_{inv}(\CStorage_{inv}^{k-1}(\CStorage))$ configuration
  $(\sigma_1, m_1) \dots (\sigma_n, m_n)$ as the
  $\CStorage_{inv}^k(\CStorage)$ configuration
  $(\bot, m_1+\sigma_1) \dots (\bot,m_n+\sigma_n)$.
  Simulation is now carried out similar to the case $k=2$. The
  simulation of the test $\TOP{\sigma}$ is by doing the right number
  of level $1$ $\pop$ operations followed by a inverse push of level
  $2$ and then again restoring the initial storage configuration.
  If we want to apply a storage operation (different from
  $\push{\sigma,\id}$ and $\invpush{\sigma}$)
  to  $(\bot, m_1+\sigma_1) \dots (\bot,m_n+\sigma_n)$, we first
  restore the configuration
  $(\bot, m_1+\sigma_1) \dots (\bot,m_n)$ then apply the configuration
  and afterwards restore the encoding of $\sigma_n$.
  For the $\push{\sigma,\id}$ we just apply $\push{\bot,\id}$ and
  subsequently replace the topmost $m_n+\sigma_n$ by $m_n+\sigma$.
  For the inverse push, we first have to guess $\sigma_{n-1}$, replace
  $m_n+\sigma_n$ by $m_n+\sigma_{n-1}$ and then apply the inverse
  push.

  Again it is straightforward to prove that this simulation is
  correct.
\end{proof}

\begin{corollary}
  $\CStorage_{inv}^k(\CStorage) \equiv_N
  \PStorage_{inv}^k(\CStorage)$.
\end{corollary}
\begin{proof}
  By induction on $k$ and Lemmas  \ref{lem:PinvtoCinvNond} and
  Proposition \ref{prop:prececNondCompatiblePStorage}, we obtain
  \begin{equation*}
    \begin{aligned}
    \PStorage_{inv}^k(\CStorage) &=
    \PStorage_{inv}(\PStorage_{inv}^{k-1}(\CStorage)) \equiv_N
    \PStorage_{inv}(\CStorage_{inv}^{k-1}(\CStorage))\\
    &\equiv_N
    \CStorage_{inv}(\CStorage_{inv}^{k-1}(\CStorage)) \equiv_N
    \CStorage_{inv}^k(\CStorage)
    \end{aligned}
  \end{equation*}
\end{proof}

\end{document}